\documentclass[12pt]{amsart}
\usepackage{amsmath}
\usepackage{amsxtra}
\usepackage{amscd}
\usepackage{amsthm}
\usepackage{amsfonts}
\usepackage{amssymb}
\usepackage{eucal}

\input prepictex
\input pictex
\input postpictex
\usepackage{mathptm}

\usepackage{verbatim}

\newtheorem{theorem}{Theorem}[section]

\newtheorem{lemma}[theorem]{Lemma}

\newtheorem{proposition}[theorem]{Proposition}

\theoremstyle{definition}
\newtheorem{definition}[theorem]{Definition}

\theoremstyle{remark}
\newtheorem{remark}[theorem]{Remark}

\numberwithin{equation}{section}

\newcommand{\R}{{\mathbb R}}
\newcommand{\C}{{\mathbb C}}

\newcommand{\barh}{{\bar h}}
\newcommand{\cA}{{\mathcal A}}
\newcommand{\cG}{{\mathcal G}}

\newcommand{\cM}{{\mathcal M}}
\newcommand{\cR}{{\mathcal R}}
\newcommand{\cE}{{\mathcal E}}

\newcommand{\bM}{{\mathbf M}}

\newcommand{\lAm}{{_l\cA^m}}
\newcommand{\rAm}{{\cA_r^m}}

\newcommand{\g}{{\mathbf g}}

\begin{document}
\title[Noncommutative spacetime]
{A noncommutative geometric approach to \\ the quantum structure of spacetime}

\author{R.B. Zhang}
\address{School of Mathematics and Statistics,
University of Sydney, Sydney, Australia}
\email{rzhang@sydney.edu.au}

\author{Xiao Zhang}
\address{Institute of Mathematics, Academy of Mathematics and Systems Science,
Chinese Academy of Sciences, Beijing, China}
\email{xzhang@amss.ac.cn}

\begin{abstract}
Together with collaborators, we introduced a noncommutative
Riemannian geometry over Moyal algebras and systematically developed
it for noncommutative spaces embedded in higher dimensions in the
last few years. The theory was applied to construct a noncommutative
version of general relativity, which is expected to capture some
essential structural features of spacetime at the Planck scale.
Examples of noncommutative spacetimes were investigated in detail.
These include quantisations of plane-fronted gravitational waves,
quantum Schwarzschild spacetime and Schwarzschild-de Sitter
spacetime, and a quantun Tolman spacetime which is relevant to
gravitational collapse.  Here we briefly review the theory and its
application in the study of quantum structure of spacetime.
\end{abstract}
 \maketitle

\tableofcontents
\section{Introduction}\label{sect:introd}
\subsection{Quantum spacetime and noncommutative geometry}

It is a common consensus in the physics community that at the Planck
scale ($1.6 \times 10^{-33}$ cm), quantum gravitational effects
become dominant, and the usual notion of spacetime as a pseudo
Riemannian manifold becomes obsolete.  In the 40s, great masters
like Heisenberg, Yang and others already pondered about the
possibility that spacetime might become noncommutative \cite{Sn, Y}
at the Planck scale, and a rather convincing argument in support of
this was given by Doplicher, Fredenhagen and Roberts \cite{DFR} in
the middle of 90s. The essence of their argument is as follows. When
one localizes spacetime events with extreme precision, gravitational
collapse will occur. Therefore, possible accuracy of localization of
spacetime events should be limited in a quantum theory incorporating
gravitation. This implies uncertainty relations for the different
coordinates of spacetime events similar to Heisenberg's uncertainty
principle in quantum mechanics. These authors then proposed a model
of quantum spacetime based on a generalisation of the Moyal algebra,
where the commutation relations among coordinates implement the
uncertainty relations.

The argument of \cite{DFR} is consistent with the current
understanding of gravitational physics. It inevitably leads to the
conclusion that some form of {\em noncommutative geometry} \cite{Co}
will be necessary in order to describe the structure of spacetime at
the Planck scale.

Early work of Yang \cite{Y} and of Snyder \cite{Sn} already
contained genesis of geometries which were not noncommutative, but
noncommutative geometry became a mathematical subject in
its own right only in the middle of 80s
following Connes' work. Since then there has been much
progress both in developing theories and exploring their
applications. Many viewpoints were adopted and different
mathematical approaches were followed by different researchers.
Connes' theory \cite{Co} (see also \cite{GVF}) formulated within the
framework of $C^*$-algebras is the most successful, which
incorporates cyclic cohomology and K-theory, and gives rise to
noncommutative versions of index theorems. A central notion in the
current formulation of Connes' theory is that of a spectral triple
involving an operator generalising the Dirac operator on a spin
manifold. It is known that one can recover spin manifolds from the
framework of spectral triples.

Theories generalizing aspects of algebraic geometry were also
developed (see, e.g., \cite{SV} for a review and references). A
notion of noncommutative schemes was formulated, which seems to
provide a useful framework for developing noncommutative algebraic
geometry. The seminal work of Kontsevich (see \cite{Ko}) on
deformation quantization \cite{BFFLS} of Poisson manifolds and
further developments \cite{KS} along a similar line are another
aspect of the subject of noncommutative geometry.

\subsection{Applications of noncommutative geometry in quantum physics}

Noncommutative geometry has been applied to several areas in quantum physics.
Connes and Lott \cite{CL} obtained the
classical Lagrangian of the standard model by ``dimensional
reduction" from a noncommutative space to the four dimensional
Minkowski space. Latter Chamseddine and Connes \cite{CC96} used the
more sophisticated method of spectral triples to incorporate
gravity, which introduces additional interaction terms to the usual
Lagrangian of standard model coupled to gravity. See \cite{CC10} for
recent developments.

Seiberg and Witten \cite{SW} showed that the anti-symmetric tensor
field arising from massless states of strings can be described by
the noncommutativity of a spacetime, where the algebra of functions
is governed by the Moyal product. A considerable amount of research
followed \cite{SW}, see the review articles \cite{DN, Sz} and
references therein.

Noncommutative quantum field theoretical models were obtained by
replacing the usual product of classical fields by the Moyal product
(see \cite{Sz} for a review and and references). Such theories were
shown to have the curious property of mixing infrared and
ultraviolet divergences in Feynman diagrams \cite{MRS}. A solid
result \cite{HW} in the area seems to be the proof of
renormalisability of a noncommutative version of the $\phi^4$ theory
in $4$-dimensions (which has an external harmonic oscillator
potential term thus manifestly breaks translational invariance).
However, many fundamental issues remain murky. For example, it was
suggested that the Poincar\'e invariance of relativistic quantum
field theory was deformed to a twisted Poincar\'e invariance
\cite{CKNT, CPT} in noncommutative field theory, but it appears
rather impossible to define actions of the twisted Poincar\'e
algebra on noncommutative fields \cite{CKTZZ}. There was even the
surprising claim that noncommutative quantum field theories were
completely equivalent to their undeformed counter parts \cite{FW}.

An important application of noncommutative geometry is in the study
of noncommutative generalisations of Einstein's theory of general
relativity. A consistent formulation of a noncommutative version of
general relativity is expected to give insight into a gravitational
theory compatible with quantum mechanics. A unification of general
relativity with quantum mechanics has long been sought after but
remains as elusive as ever despite the extraordinary efforts put
into string theory for the last three decades. The noncommutative
geometrical approach may provide an alternative route. There have
been intensive research activities in this general direction
inspired by \cite{DFR}. We refer to \cite{MH} for a brief review.
More references can also be found in \cite{Sz}.

\subsection{A noncommutative geometric approach to quantum
spacetime}\label{sect:CTWZZ}

In \cite{CTZZ, ZZ}, a noncommutative Riemannian geometry over Moyal
algebras was developed, which retains key notions such as metric,
connection and curvature of usual Riemannian geometry. The theory
was applied to develop a noncommutative theory of general
relativity, which is expected to capture essential features of
quantum gravity.

We first quantised a space by deforming \cite{Ge, Ko} the algebra of
functions to a noncommutative associative algebra, the Moyal
algebra. Such an algebra naturally incorporates the generalised
spacetime uncertainty relations of \cite{DFR}, capturing key
features expected of spacetime at the Planck scale. We then
systematically investigated the noncommutative Riemannian geometry
of noncommutative spaces embedded in ``higher dimensions". The
general theory was extracted from the noncommutative Riemannian
geometry of embedded spaces.

The theory of \cite{CTZZ, ZZ} was first developed within a geometric
framework analogous to the classical theory of embedded surfaces
(see, e.g., \cite{doC}). This has the advantage of being concrete
and explicit. Many examples of such noncommutative geometries can be
easily constructed, which are transparently consistent in contrast.
We then reformulated the theory algebraically \cite{ZZ} in terms of
projective modules, a language commonly adopted in noncommutative
geometry \cite{Co, GVF}. Morally a projective module (a direct
summand of a free module) is the geometric equivalent of an
embedding of a low dimensional manifold isometrically in a higher
dimensional one.

This connection between the algebraic notion of projective modules
and the geometric notion of embeddings is particularly significant
in view of Nash's isometric embedding theorem \cite{N} and its
generalisation to pseudo-Riemannian manifolds \cite{Fr, C, Gr}. The
theorems state that any (pseudo-) Riemannian manifold can be
isometrically embedded in flat spaces. Therefore, in order to study
the geometry of spacetime, one only needs to investigate (pseudo-)
Riemannian manifolds embedded in higher dimensions. It is reasonable
to anticipate something similar in the noncommutative setting.
We should mention that embedded noncommutative spaces also
play a role in the study of branes embedded in $\R^D$ in the context
of Yang-Mills matrix models \cite{St}.

The theory of \cite{CTZZ, ZZ} was applied to construct a
noncommutative analogue of general relativity. In particular, a
noncommutative Einstein field equation was proposed based on
analysis of the noncommutative analogue of the second Bianchi
identity. A new feature of the noncommutative field equation is the
presence of another noncommutative Ricci curvature tensor (see
\eqref{Einstein}). The highly nonlinear nature of the equation makes
it difficult to study, and noncommutativity adds
further complication. Despite this, we have succeeded in
constructing a class of exact solutions of the noncommutative
Einstein equation, which are quantum deformations of plane-fronted
gravitational waves \cite{Br, ER, Ro, EK}.

Quantisations of the Schwarzschild spacetime and Schwarzschild-de
Sitter spacetime were also investigated within the framework of
\cite{CTZZ, ZZ}. They solve the noncommutative Einstein equation in
the vacuum to the first order in the deformation parameter. However,
higher order terms appear, which may be interpreted as matter
sources. The physical origin and implications of the source terms
are issues of physical interest. The Hawking temperature and entropy
of the quantum Schwarzschild black hole were investigated, and a
quantum correction to the entropy-area law was observed.

Gravitational collapse was also studied in the noncommutative
setting \cite{SWXZZ}. Classically one may follow \cite{W} to
investigate the behaviour of the scalar curvature of the Tolman
spacetime. When time increases to a certain critical value, the
scalar curvature goes to infinity, thus the radius of the stellar
object reduces to zero. By the reasoning of \cite{W}, this indicates
gravitational collapse. [Obviously this only provides a snapshot of
the evolution of the star, nevertheless, it enables one to gain some
understanding of gravitational collapse.] A similar analysis in the
noncommutative setting showed that gravitational collapse happens
within a certain range of time instead of a single critical value,
because of the quantum effects captured by the non-commutativity of
spacetime. However, noncommutativity effect enters only at third
order of the deformation parameter.

\subsection{The present paper}

This paper is a mathematical review of the theory developed in
\cite{CTZZ, ZZ} and its application to noncommutative gravity
\cite{WZZ1, WZZ2, SWXZZ}. Its main body consists of two parts.
One part comprises of Sections
\ref{sect:Moyal}--\ref{sect:transformations}, which are all on the
general theory, except for Section \ref{sect:examples}, where some
elementary examples of noncommutative spaces are worked out in
detail to illustrate the general theory. The other part comprises of
Sections \ref{sect:Einstein}--\ref{sect:spacetime}. In Section
\ref{sect:Einstein}, we propose noncommutative
Einstein field equations, and construct exact solutions for them in the
vacuum. In section \ref{sect:spacetime}, we discuss the structure of
the quantum versions of the Schwarzschild spacetime,
Schwarzschild-de Sitter spacetime, and the Tolman spacetime which is
relevant for gravitational collapse.

We emphasize that the present paper is nothing more than a
streamlined account of the works \cite{CTZZ, ZZ, WZZ1, WZZ2, SWXZZ}.
It is certainly {\em not} meant to be a review of noncommutative
general relativity. There is a vast body of literature in this
subject area. The time and energy required to filter the literature
to write an in-depth critique are beyond our means. As a result, we
did not make any attempt to include all the references.

Here we merely comment that a variety of physically motivated
methods and techniques were used to study corrections to general
relativity arising from the noncommutativity of the Moyal algebra.
For example, references \cite{ADMW1, ADMW2} studied deformations of
the diffeomorphism algebra as a means for incorporating
noncommutative effects of spacetime. [It is unfortunate that
gravitational theories proposed this way \cite{ADMW1, ADMW2} were
different  \cite{AMV} from the low energy limit of string theory.]
In \cite{CTZ, COTZ} a gauge theoretical approached was taken. These
approaches differ considerably from the theory of \cite{CTZZ, WZZ1,
WZZ2} mathematically. Other types of noncommutative Riemannian
geometries were also proposed \cite{DMMM, DHLS, MM, M05}, which
retain some of the familiar geometric notions like metric and
curvature. Noncommutative analogues of the Hilbert-Einstein action
were also suggested \cite{C01, C04} by treating noncommutative
gravity as gauge theories. In the last couple of years, there have
also been many other papers on noncommutative black holes, see e.g.,
\cite{CTZ, ANSS, DKS, BNK, Kob}. We should also mention that Moyal planes
have been treated from the point of view of spectral triples in \cite{GBV}.

\medskip

\noindent{\bf Acknowledgement.} We thank Masud Chaichian, Anca
Tureanu, Ding Wang, Wen Sun and Naqing Xie for collaborations
at various stages of the project on noncommutative gravity.

\section{Differential geometry of noncommutative vector bundles}\label{sect:bundles}

In this section we investigate general aspects of the noncommutative
differential geometry over the Moyal algebra. We shall focus on the
abstract theory here. A large class of examples will be given in
later sections.

\subsection{Moyal algebra and projective modules}\label{sect:Moyal}
We recall the definition of the Moyal algebra of smooth
functions on an open region of $\R^n$, and also describe the
finitely generated projective modules over the Moyal algebra. This
provides the background material needed, and also serves to fix
notation.

We take an open region $U$ in $\R^n$ for a fixed $n$, and write the
coordinate of a point $t\in U$ as $(t^1, t^2, \dots, t^n)$. Let
$\barh$ be a real indeterminate, and denote by $\R[[\barh]]$ the
ring of formal power series in $\barh$. Let $\cA$ be the set of the
formal power series in $\barh$ with coefficients being real smooth
functions on $U$. Namely, every element of $\cA$ is of the form
$\sum_{i\ge 0} f_i\barh^i$ where $f_i$ are smooth functions on $U$.
Then $\cA$ is an $\R[[\barh]]$-module in the obvious way.

Fix a constant skew symmetric $n\times n$ matrix $\theta=(\theta_{i
j})$. The Moyal product on $\cA$ corresponding to $\theta$ is a map
\[ \mu: \cA \otimes_{\R[[\barh]]} \cA \longrightarrow \cA,
\quad f\otimes g \mapsto \mu(f, g),
\]
defined by
\begin{eqnarray}\label{multiplication}
\mu(f, g)(t)= \lim_{t'\rightarrow t} \exp^{\barh \sum_{i, j}
\theta_{i j}\frac{\partial}{\partial t^i}\frac{\partial}{\partial
t'^j}}f(t) g(t').
\end{eqnarray}
On the right hand side,  $f(t) g(t')$ means the usual product of the
functions $f$ and $g$ at $t$ and $t'$ respectively. Here
$\exp^{\barh \sum_{i j} \theta_{i j}\frac{\partial}{\partial
t^i}\frac{\partial}{\partial t'^j}}$ should be understood as a power
series in the differential operator $\sum_{i, j} \theta_{i
j}\frac{\partial}{\partial t^i}\frac{\partial}{\partial t'^j}$. We
extend $\mu$ $\R[[\barh]]$-linearly to all elements in $\cA$ by
letting
\begin{eqnarray*}
\mu(\sum f_i \barh^i, \sum g_j \barh^j) :=\sum  \mu(f_i, g_j)
\barh^{i+j}.
\end{eqnarray*}

It has been known since the late 40s from work of J. E. Moyal that
the Moyal product is associative. Thus the $\R[[\barh]]$-module
$\cA$ equipped with the Moyal product forms an associative algebra
over $\R[[\barh]]$, which is a deformation of the algebra of smooth
functions on $U$ in the sense of \cite{Ge}. We shall usually denote
this associative algebra by $\cA$, but when it is necessary to make
explicit the multiplication, we shall write it as $(\cA, \mu)$.

The partial derivations $\partial_i:=\frac{\partial}{\partial t^i}$
($i=1, 2, \dots, n$) with respect to the coordinates $t^i$ for $U$
are $\R[[\barh]]$-linear maps on $\cA$. Since $\theta$ is a constant
matrix, we have the Leibniz rule
\begin{eqnarray}\label{Leibniz}
\partial_i\mu(f, g)= \mu(\partial_i f, g) + \mu(f, \partial_i g)
\end{eqnarray}
for any elements $f$ and $g$ of $\cA$. Therefore, the $\partial_i$
are mutually commutating derivations of the Moyal algebra $(\cA,
\mu)$.

\begin{remark}
The usual notation in the literature for $\mu(f, g)$ is $f\ast g$,
which is also referred to as the star-product of $f$ and $g$.
Hereafter we shall replace $\mu$ by $\ast$ and simply write $\mu(f,
g)$ as $f\ast g$.
\end{remark}

Following the general philosophy of noncommutative geometry
\cite{Co}, we regard the noncommutative associative algebra $(\cA,
\mu)$ as defining some {\em quantum deformation of the region $U$}.
Finitely generated projective modules over $\cA$ will be regarded as
(spaces of sections of) noncommutative vector bundles on the quantum
deformation of $U$ defined by $\cA$. Let us now briefly describe
finitely generated projective $\cA$-modules.

Given an integer $m>n$, we let $\lAm$ (resp. $\rAm$) be the set of
$m$-tuples with entries in $\cA$ written as rows (resp. columns). We
shall regard $\lAm$ (resp. $\rAm$) as a left (resp. right)
$\cA$-module with the action defined by multiplication from the left
(resp. right). More explicitly, for $v=\begin{pmatrix}a_1 & a_2 &
\dots & a_m\end{pmatrix}\in\lAm$, and $b\in\cA$, we have $b \ast v =
\begin{pmatrix}b\ast a_1 & b\ast a_2 & \dots & b\ast a_m\end{pmatrix}$.
Similarly for $w=\begin{pmatrix}a_1 \\  a_2 \\ \vdots \\
a_m\end{pmatrix}\in\rAm$, we have $w\ast b = \begin{pmatrix}a_1\ast b \\
a_2\ast b\\ \vdots \\ a_m\ast b\end{pmatrix}$. Let $\bM_m(\cA)$ be
the set of $m\times m$-matrices with entries in $\cA$. We define
matrix multiplication in the usual way but by using the Moyal
product for products of matrix entries, and still denote the
corresponding matrix multiplication by $\ast$. Now for $A=(a_{i j})$
and $B=(b_{i j})$, we have $(A\ast B)=(c_{i j})$ with $c_{i j} =
\sum_{k} a_{i k}\ast b_{k j}$. Then $\bM_m(\cA)$ is an
$\R[[\barh]]$-algebra, which has a natural left (resp. right) action
on $\rAm$ (resp. $\lAm$).

A finitely generated projective left (reps. right) $\cA$-module is
isomorphic to some direct summand of $\lAm$ (resp. $\rAm$) for some
$m<\infty$. If $e\in \bM_m(\cA)$ satisfies the condition $e\ast
e=e$, that is, it is an idempotent, then
\[
\cM=\lAm\ast e := \{v\ast e \mid v\in \lAm \}, \quad
\tilde\cM=e\ast\rAm := \{e\ast w \mid \in \rAm \}
\]
are respectively projective left and right $\cA$-modules.
Furthermore, every projective left (right) $\cA$-module is
isomorphic to an $\cM$ (resp. $\tilde\cM$) constructed this way by
using some idempotent $e$.

\subsection{Connections and curvatures}\label{connections}

We start by considering the action of the partial derivations
$\partial_i$ on $\cM$ and $\tilde\cM$. We only treat the left module
in detail, and present the pertinent results for the right module at
the end, since the two cases are similar.

Let us first specify that $\partial_i$ acts on rectangular matrices
with entries in $\cA$ by componentwise differentiation. More
explicitly,
\[
\partial_i B=
\begin{pmatrix} \partial_i b_{1 1} & \partial_i b_{1 2} & \dots &\partial_i b_{1 l}\\
                \partial_i b_{2 1} & \partial_i b_{2 2} & \dots &\partial_i b_{2 l}\\
                  \dots   & \dots    & \dots &\dots\\
                \partial_i b_{k 1} & \partial_i b_{k 2} & \dots &\partial_i b_{k l}
\end{pmatrix} \quad \text{for} \quad B=
\begin{pmatrix}   b_{1 1} & b_{1 2} & \dots &b_{1 l}\\
                  b_{2 1} & b_{2 2} & \dots &b_{2 l}\\
                  \dots   & \dots    & \dots &\dots\\
                  b_{k 1} & b_{k 2} & \dots &b_{k l}
\end{pmatrix}.
\]
In particular, given any $\zeta=v\ast e\in \cM$, where $v\in\lAm$
regarded as a row matrix, we have $\partial_i\zeta = (\partial_i
v)\ast e + v\ast\partial_i(e)$ by the Leibniz rule. While the first
term belongs to $\cM$, the second term does not in general.
Therefore, $\partial_i$ ($i=1, 2, \dots, n$) send $\cM$ to some
subspace of $\lAm$ different from $\cM$.

Let $\omega_i \in \bM_m(\cA)$ ($i=1, 2, \dots, n$) be $m\times
m$-matrices with entries in $\cA$ satisfying the following
condition:
\begin{eqnarray} \label{connect-property}
e\ast\omega_i\ast(1-e) = -e\ast\partial_i e, \quad \forall i.
\end{eqnarray}
Define the $\R[[\barh]]$-linear maps $\nabla_i$ ($i=1, 2, \dots, n$)
from $\cM$ to $\lAm$ by
\[
\nabla_i \zeta = \partial_i\zeta + \zeta \ast\omega_i , \quad
\forall \zeta \in \cM.
\]
Then each $\nabla_i$ is a covariant derivative on the noncommutative
bundle $\cM$ in the sense of Theorem \ref{key} below. They together
define a {\em connection} on $\cM$.
\begin{theorem}\label{key}
The maps $\nabla_i$ ($i=1, 2, \dots, n$) have the following
properties. For all $\zeta\in \cM$ and $a\in \cA$,
\[\nabla_i\zeta \in \cM \quad  \text{and} \quad
\nabla_i(a \ast \zeta)=\partial_i(a)\ast \zeta + a \ast\nabla_i\zeta.\]
\end{theorem}
\begin{proof}
For any $\zeta\in \cM$, we have
\[
\begin{aligned}
\nabla_i(\zeta)\ast e   &=\partial_i(\zeta)\ast e + \zeta \ast\omega_i\ast e\\
                    &= \partial_i\zeta + \zeta \ast(\omega_i\ast e -
                    \partial_i e),
\end{aligned}
\]
where we have used the Leibniz rule and also the fact that $\zeta
\ast e = \zeta$. Using this latter fact again, we have
$ \zeta\ast (\omega_i\ast e - \partial_i e) =  \zeta\ast
(e\ast\omega_i\ast e -e\ast\partial_i e) $, and by the defining
property \eqref{connect-property} of $\omega_i$, we obtain $\zeta
\ast(e\ast\omega_i\ast e -e\ast\partial_i\ast e) = \zeta
\ast\omega_i$. Hence
\[
\nabla_i(\zeta)\ast  e = \partial_i\zeta + \zeta\ast \omega_i
=\nabla_i\zeta,
\]
proving that $\nabla_i\zeta\in\cM$. The second part of the theorem
immediately follows from the Leibniz rule.
\end{proof}

We shall also say that the set of $\omega_i$ ($i=1, 2, \dots, n$) is
a connection on $\cM$. Since $e\ast\partial_i e =
\partial_i(e)\ast(1-e)$, one obvious choice for $\omega_i$ is $\omega_i
= - \partial_i e$, which we shall refer to as the {\em canonical
connection} on $\cM$.

The following result is an easy consequence of \eqref{connect-property}.
\begin{lemma}
If $\omega_i$ ($i=1, 2, \dots, n$) define a connection on $\cM$,
then so do also $\omega_i + \phi_i\ast e $ ($i=1, 2, \dots, n$) for
any $m\times m$-matrices $\phi_i$ with entries in $\cA$.
\end{lemma}

For a given connection $\omega_i$ ($i=1, 2, \dots, n$), we consider
$[\nabla_i, \nabla_j]= \nabla_i \nabla_j -  \nabla_j \nabla_i$ with
the right hand side understood as composition of maps on $\cM$.
Simple calculations show that for all $\zeta\in\cM$,
\[
\begin{aligned}
{[\nabla_i, \nabla_j]}\zeta = \zeta\ast \cR_{i j} \quad \text{with}
\quad \cR_{i j}:=\partial_i \omega_j - \partial_j \omega_i -
[\omega_i, \ \omega_j]_\ast,
\end{aligned}
\]
where $[\omega_i, \omega_j]_\ast = \omega_i\ast \omega_j
-\omega_j\ast \omega_i$ is the commutator. We call $\cR_{i j}$ the
{\em curvature} of $\cM$ associated with the connection $\omega_i$.

For all $\zeta\in\cM$,
\[
\begin{aligned}
{[\nabla_i, \nabla_j]}\nabla_k \zeta &= \partial_k(\zeta) \ast
\cR_{i j}
+ \zeta \ast\omega_k\ast\cR_{i j}, \\
\nabla_k[\nabla_i, \nabla_j] \zeta &= \partial_k(\zeta)\ast \cR_{i
j} + \zeta\ast (\partial_k\cR_{i j} + \cR_{i j}\ast \omega_k).
\end{aligned}
\]
Define the following covariant derivatives of the curvature:
\begin{eqnarray}
\nabla_k \cR_{i j} := \partial_k\cR_{i j} + \cR_{i j}\ast \omega_k -
\omega_k\ast \cR_{i j},
\end{eqnarray}
we have
\[
{[\nabla_k, [\nabla_i, \nabla_j]]} \zeta = \zeta \ast\nabla_k \cR_{i
j}, \quad \forall \zeta\in\cM.
\]
The Jacobian identity $[\nabla_k, [\nabla_i, \nabla_j]]+[\nabla_j,
[\nabla_k, \nabla_i]]+[\nabla_i, [\nabla_j, \nabla_k]]=0$ leads to
\[
\zeta\ast (\nabla_k \cR_{i j}+\nabla_j \cR_{k i}+\nabla_i \cR_{j
k})=0, \quad \forall \zeta\in\cM.
\]
From this we immediately see that $e\ast(\nabla_k \cR_{i j}+\nabla_j
\cR_{k i}+\nabla_i \cR_{j k})=0$. In fact, the following stronger
result holds.
\begin{theorem} The curvature satisfies the following {\em Bianchi identity}:
\[
\nabla_k \cR_{i j}+\nabla_j \cR_{k i}+\nabla_i \cR_{j k}=0.
\]
\end{theorem}
\begin{proof}
The proof is entirely combinatorial. Let
\[
\begin{aligned}
A_{i j k} &= \partial_k\partial_i\omega_j -
\partial_k\partial_j\omega_i, \\
B_{i j k} &= [\partial_i\omega_j, \omega_k]_\ast -
[\partial_j\omega_i, \omega_k]_\ast.
\end{aligned}
\]
Then we can express $\nabla_k\cR_{i j}$ as
\[\nabla_k\cR_{i j} = A_{i j k} + B_{i j k} - \partial_k[\omega_i,\
\omega_j]_\ast - [[\omega_i,\ \omega_j]_\ast, \ \omega_k]_\ast.
\]
Note that
\[
\begin{aligned}
A_{i j k} + A_{j k i} + A_{k i j}&=0, \\
B_{i j k} + B_{j k i} + B_{k i j}&= \partial_k[\omega_i,\
\omega_j]_\ast + \partial_i[\omega_j,\
\omega_k]_\ast+\partial_j[\omega_k,\ \omega_i]_\ast.
\end{aligned}
\]
Using these relations together with the Jacobian identity
\[
[[\omega_i,\ \omega_j]_\ast, \ \omega_k]_\ast+[[\omega_j,\
\omega_k]_\ast, \ \omega_i]_\ast+[[\omega_k,\ \omega_i]_\ast, \
\omega_j]_\ast=0,
\]
we easily prove the Bianchi identity.
\end{proof}

\subsection{Gauge transformations}\label{subsect-gauge}

Let $GL_m(\cA)$ be the group of invertible $m\times m$-matrices with
entries in $\cA$. Let $\cG$ be the subgroup defined by
\begin{eqnarray}\label{gauge-group}
\cG=\{ g\in GL_m(\cA) \mid e\ast g = g\ast e\},
\end{eqnarray}
which will be referred to as the {\em gauge group}. There is a right
action of $\cG$ on $\cM$ defined, for any $\zeta\in \cM$ and
$g\in\cG$, by $\zeta \times g \mapsto \zeta\cdot g:=\zeta \ast g$,
where the right side is defined by matrix multiplication. Clearly,
$\zeta \ast g \ast e = \zeta \ast g$. Hence $\zeta \ast g\in\cM$,
and we indeed have a $\cG$ action on $\cM$.

For a given $g\in\cG$, let
\begin{eqnarray} \label{connect-gauged}
\omega_i^g =g^{-1}\ast\omega_i\ast g - g^{-1}\ast\partial_i g.
\end{eqnarray}
Then
\[
\begin{aligned}
e\ast \omega_i^g\ast (1-e) &= g^{-1}\ast e\ast \omega_i
\ast(1-e)\ast g - g^{-1}\ast e\ast
\partial_i(g)\ast (1-e).
\end{aligned}
\]
By \eqref{connect-property}, $g^{-1}\ast e\ast\omega_i\ast (1-e)\ast g = - g^{-1}\ast e\ast\partial_i(e)\ast g$.
Using the defining property of the gauge group $\cG$, we can show that
\[
\begin{aligned}
g^{-1}\ast e\ast\omega_i\ast (1-e)\ast g
&=- e\ast\partial_i e  +  g^{-1}\ast e \ast\partial_i(g)\ast(1-e).
\end{aligned}
\]
Therefore,
$
e\ast\omega_i^g \ast(1-e) = - e\ast\partial_i e.
$
This shows that the $\omega_i^g$ satisfy the condition
\eqref{connect-property}, thus form a connection on $\cM$.

Now for any given $g\in\cG$,  define the maps $\nabla_i^g$ on $\cM$
by
\[
\nabla_i^g\zeta =\partial_i\zeta + \zeta\ast \omega_i^g, \quad
\forall \zeta.
\]
Also, let $\cR_{i j}^g=\partial_i \omega_j^g - \partial_j \omega_i^g
-[\omega_i^g, \, \omega_j^g]_\ast$ be the curvature corresponding to
the connection $\omega_i^g$. Then we have the following result.
\begin{lemma}\label{gauge} Under a gauge transformation procured by
$g\in\cG$,
\[
\begin{aligned} &\nabla_i^g(\zeta\ast g) =\nabla_i(\zeta)\ast g, \quad
\forall \zeta\in\cM; \\ &\cR_{i j}^g = g^{-1}\ast \cR_{i j} \ast g.
\end{aligned}
\]
\end{lemma}
\begin{proof}
Note that
\[
\nabla_i^g(\zeta\ast g) = \partial_i(\zeta)\ast g + \zeta\ast
\partial_i g + \zeta\ast g\ast \omega_i^g = (\partial_i\zeta +\zeta\ast
\omega_i)\ast g.
\]
This proves the first formula.

To prove the second claim, we use the following formulae
\[
\begin{aligned}
\partial_i \omega_j^g - \partial_j \omega_i^g
&= g^{-1}\ast ( \partial_i \omega_j - \partial_j \omega_i)\ast g -
\partial_i(g^{-1}) \ast\partial_j g + \partial_j(g^{-1}) \ast\partial_i g\\
& + [\partial_i(g^{-1})\ast g,\ g^{-1}\ast\omega_j\ast g]_\ast -
[\partial_j (g^{-1})\ast g,\ g^{-1}\ast\omega_i\ast g]_\ast; \\
[\omega_i^g, \ \omega_j^g]_\ast&= g^{-1}\ast[\omega_i, \
\omega_j]_\ast \ast g -
\partial_i(g^{-1}) \ast\partial_j g + \partial_j(g^{-1})\ast \partial_i g\\
& + [\partial_i(g^{-1})\ast g,\ g^{-1}\ast\omega_j\ast g]_\ast -
[\partial_j (g^{-1})\ast g,\ g^{-1}\ast\omega_i\ast g]_\ast.
\end{aligned}
\]
Combining these formulae together we obtain $\cR_{i j}^g = g^{-1}
\cR_{i j} g$. This completes the proof of the lemma.
\end{proof}

\subsection{Vector bundles associated to right projective modules}

Connections and curvatures can be introduced for the right bundle
$\tilde\cM=e\ast\rAm$ in much the same way. Let
$\tilde\omega_i\in\bM_m(\cA)$ ($i=1, 2, \dots, n$) be matrices
satisfying the condition that
\begin{eqnarray}\label{connection-right}
(1-e)\ast\tilde\omega_i\ast e = \partial_i(e)\ast e.
\end{eqnarray}
Then we can introduce a connection consisting of the right covariant
derivatives $\tilde\nabla_i$ ($i=1, 2, \dots, n$) on $\tilde\cM$
defined by
\[
\begin{aligned}
\tilde\nabla_i: \tilde\cM \longrightarrow \tilde\cM,  &\quad &\xi
\mapsto \tilde\nabla_i\xi = \partial_i\xi - \tilde\omega_i\ast\xi.
\end{aligned}
\]
It is easy to show that $\tilde\nabla_i(\xi \ast a) =
\tilde\nabla_i(\xi)\ast a + \xi\ast \partial_ia$ for all $a\in\cA$.

Note that if $\tilde\omega_i=\partial_i e$ for all
$i$, the condition \eqref{connection-right} is satisfied. We call
them  the {\em canonical connection} on $\tilde\cM$.

Returning to a general connection $\tilde\omega_i$, we define the
associated curvature by
\[
\tilde\cR_{i j} = \partial_i\tilde\omega_j -
\partial_j\tilde\omega_i -[\tilde\omega_i, \ \tilde\omega_j]_\ast.
\]
Then for all $\xi\in\tilde\cM$, we have
\[ [\tilde\nabla_i, \ \tilde\nabla_j]\xi = -\tilde\cR_{i j}\ast\xi. \]
We further define the covariant derivatives of $\tilde\cR_{i j}$ by
\[ \tilde\nabla_k\tilde\cR_{i j} = \partial_k \tilde\cR_{i j} +
\tilde\omega_k \ast\tilde\cR_{i j}- \tilde\cR_{i j}\ast
\tilde\omega_k.
\]
Then we have the following result.
\begin{lemma} The curvature on the right bundle $\tilde\cM$ satisfies the
Bianchi identity
\[\tilde\nabla_i\tilde\cR_{j k} +
\tilde\nabla_j\tilde\cR_{k i}+ \tilde\nabla_k\tilde\cR_{i j}=0. \]
\end{lemma}
By direct calculations we can also prove the following result:
\[
[\tilde\nabla_k, [\tilde\nabla_i, \ \tilde\nabla_j]]\xi
=-\tilde\nabla_k(\tilde\cR_{i j})\ast \xi, \quad \forall
\xi\in\tilde\cM.
\]

Consider the gauge group $\cG$ defined by \eqref{gauge-group}, which
has a right action on $\tilde\cM$:
\[
\tilde\cM \times \cG \longrightarrow \tilde\cM, \quad \xi\times g
\mapsto \xi\cdot g := g^{-1}\ast \xi.
\]
Under a gauge transformation procured by $g\in\cG$,
\[
\tilde\omega_i \mapsto \tilde\omega_i^g := g^{-1} \ast
\tilde\omega_i \ast g +
\partial_i(g^{-1})  \ast g.
\]
The connection $\tilde\nabla_i^g$ on $\tilde\cM$ defined by
\[ \tilde\nabla_i^g\xi = \partial_i\xi - \tilde\omega_i^g \ast\xi \]
satisfies the following relation for all $\xi\in\tilde\cM$:
\[
\tilde\nabla_i^g(g^{-1}\ast\xi) = g^{-1}\ast\tilde\nabla_i\xi.
\]
Furthermore, the gauge transformed curvature
\[
\tilde\cR_{i j}^g := \partial_i \tilde\omega_j^g - \partial_j
\tilde\omega_i^g - [\tilde\omega_i^g, \tilde\omega_j^g]_\ast
\]
is related to $\tilde\cR_{i j}$ by
\[ \tilde\cR_{i j}^g = g^{-1}\ast \tilde\cR_{i j}\ast g.\]

Given any $\Lambda\in\bM_m(\cA)$, we can define the $\cA$-bimodule
map
\begin{eqnarray}\label{form}
\langle \ , \ \rangle:
\cM\otimes_{\R[[\barh]]}\tilde\cM\longrightarrow \cA, \quad
\zeta\otimes\xi \mapsto \langle \zeta , \xi \rangle
=\zeta\ast\Lambda\ast\xi,
\end{eqnarray}
where $\zeta \ast\Lambda \ast\xi$ is defined by matrix
multiplication. We shall say that the bimodule homomorphism is {\em
gauge invariant} if for any element $g$ of the gauge group $\cG$,
\[
\langle \zeta\cdot g, \xi\cdot g \rangle = \langle \zeta, \xi
\rangle, \quad \forall \zeta\in\cM, \ \xi\in\tilde\cM.
\]
Also, the bimodule homomorphism is said to be {\em compatible with
the connections} $\omega_i$ on $\cM$ and $\tilde\omega_i$ on
$\tilde\cM$ if for all $i=1, 2, \dots, n$
\[
\partial_i \langle \zeta, \xi
\rangle = \langle \nabla_i\zeta, \xi \rangle + \langle \zeta,
\tilde\nabla_i\xi \rangle, \quad \forall \zeta\in\cM, \
\xi\in\tilde\cM.
\]

\begin{lemma}\label{compatibility}
Let $\langle \ , \ \rangle:
\cM\otimes_{\R[[\barh]]}\tilde\cM\longrightarrow \cA$ be an
$\cA$-bimodule homomorphism defined by \eqref{form} with a given
$m\times m$-matrix $\Lambda$ with entries in $\cA$. Then
\begin{enumerate}
\item \label{form-gauge} $\langle \ , \ \rangle$ is gauge invariant if
$g\ast\Lambda\ast g^{-1} = \Lambda$ for all $g\in\cG$;
\item \label{form-gauge2} $\langle \ , \ \rangle$ is compatible with the connections
$\omega_i$ on $\cM$ and $\tilde\omega_i$ on $\tilde\cM$ if for all
$i$,
\[
e\ast(\partial_i\Lambda - \omega_i \ast\Lambda +  \Lambda
\ast\tilde\omega_i)\ast e =0.
\]
\end{enumerate}
\end{lemma}
\begin{proof}
Note that $\langle \zeta\cdot g , \xi\cdot g \rangle  = \zeta \ast
g\ast \Lambda\ast g^{-1}\ast \xi$ for any $g\in\cG$, $\zeta\in\cM$
and $\xi\in\tilde\cM$. Therefore $\langle \zeta\cdot g , \xi\cdot g
\rangle = \langle \zeta, \xi\rangle$ if $g \ast\Lambda \ast
g^{-1}=\Lambda$. This proves part (\ref{form-gauge}).

Now
$
\partial_i \langle \zeta, \xi\rangle = \langle \partial_i\zeta,
\xi\rangle + \langle\zeta,  \partial_i\xi\rangle +
\zeta\ast(\partial_i\Lambda - \omega_i\ast\Lambda +
\Lambda\ast\tilde\omega_i)\ast\xi. $
Thus if $\Lambda$ satisfies the condition of part
(\ref{form-gauge2}), then $\langle \ , \ \rangle$ is compatible with
the connections.
\end{proof}

%
\subsection{Canonical connections and fibre metric}\label{canonical}

Let us consider in detail the canonical connections on $\cM$ and
$\tilde\cM$ given by
\[ \omega_i = - \partial_i e, \quad \tilde\omega_i=\partial_i e. \]
A particularly nice feature in this case is that the corresponding
curvatures on the left and right bundles coincide. We have the
following formula:
\begin{eqnarray}\label{R:leqr}
\cR_{i j} = \tilde\cR_{i j}= - [\partial_i e, \ \partial_j e]_\ast.
\end{eqnarray}

Now we consider a special case of the $\cA$-bimodule map defined by
equation \eqref{form}.
\begin{definition}\label{fibre-metric}
Denote by $\g: \cM\otimes_{\R[[\barh]]} \tilde\cM \longrightarrow
\cA$ the map defined by \eqref{form} with $\Lambda$ being the
identity matrix. We shall call $\g$ the {\em fibre metric} on $\cM$.
\end{definition}
\begin{lemma}\label{metric}
The fibre metric $\g$ is gauge invariant and is compatible with the
standard connections.
\end{lemma}
\begin{proof}
Since $\Lambda$ is the identity matrix in the present case, it
immediately follows from Lemma \ref{compatibility}
(\ref{form-gauge}) that $\g$ is gauge invariant. Note that
$e\ast\partial_i(e)\ast e =0$ for all $i$. Using this fact in Lemma
\ref{compatibility} (\ref{form-gauge2}), we easily see that $\g$ is
compatible with the standard connections.
\end{proof}

\section{Embedded noncommutative spaces}\label{sect:surfaces}

In \cite{CTZZ}, we introduced noncommutative spaces which are
embedded in ``higher dimensions" in a geometric setting. Here we
recall this theory and reformulate it in the framework of Section
\ref{sect:bundles} in terms of projective modules. This also
provides a class of explicit examples of idempotents and related
projective modules.

\subsection{Embedded noncommutative spaces}

We shall consider only embedded spaces with Euclidean signature.
The Minkowski case is similar, which will be briefly discussed in
Section \ref{Mink}. Given
$X=\begin{pmatrix} X^1 & X^2 &\dots & X^m
\end{pmatrix}$ in $\lAm$, we define an $n\times n$ matrix $\g=(g_{i
j})_{i, j=1, 2, \dots, n}$ with entries given by
\[
g_{i j}=\sum_{\alpha=1}^m \partial_i X^\alpha \ast \partial_j
X^\alpha.
\]
Let $g^0 = \g\mod\barh$, which is an $n\times n$-matrix of smooth
functions on $U$. Assume that $g^0$ is invertible at every point $t\in
U$. Then there exists a unique $n\times n$-matrix $\left(g^{i
j}\right)$ over $\cA$ which is the right inverse of $\g$, i.e.,
\[
g_{i j}\ast g^{j k} =\delta_i^k,
\]
where we have used Einstein's convention of summing
over repeated indices. To see this, we need to examine
the definition \eqref{multiplication} of the Moyal product more carefully.
Let
\begin{eqnarray}\label{mup}
\mu_p: \cA/\barh\cA\otimes\cA/\barh\cA\longrightarrow\cA/\barh\cA,
\quad p=0, 1, 2, \dots,
\end{eqnarray}
be $\R$-linear maps defined by
\[
\mu_p(f, g) = \lim_{t'\rightarrow t} \frac{1}{p!}\left(\sum_{i j}
\theta_{i, j}\frac{\partial}{\partial t^i}\frac{\partial}{\partial
t'^j}\right)^p f(t) g(t').
\]
Then $f\ast g = \sum_{p=0}^\infty \barh^p \mu_p(f, g)$.
Now write $g_{i j} =\sum_p \barh^p g_{i j}[p]$ and $g^{i j} =\sum_p
\barh^p \tilde g^{i j}[p]$, where $(g^{i j}[0])$ is the inverse of
$(g_{i j}[0])$. Now in terms of the maps $\mu_k$ defined by
\eqref{mup}, we have
\begin{eqnarray*}
\delta_i^k &=g_{i j}\ast g^{j k} &= \sum_q \barh^q
\sum_{m+n+p=q}\mu_p(g_{i j}[m], g^{j k}[n]),
\end{eqnarray*}
which is equivalent to
\begin{eqnarray*}
g^{i j}[q] &=& -\sum_{n=1}^q\sum_{m=0}^{q-n}g^{i k}[0] \mu_n(g_{k
l}[m], \, g^{l j}[q-n-m]).
\end{eqnarray*}
Since the right-hand side involves only $g^{l j}[r]$ with $r<q$,
this equation gives a recursive formula for the right inverse of
$\g$. In the same way, we can also show that there also exists a
unique left inverse of $\g$. It follows from the associativity of
multiplication of matrices over any associative algebra that the
left and right inverses of $\g$ are equal.

It is easy to see that if $\g$ is invertible,
then $g^0$ is nonsingular.

\begin{definition}\label{def:nc-space}
We call an element $X\in {_l\cA^m}$ an {\em embedded noncommutative space} if $g^0$ is invertible
for all $t\in U$. In this case, $\g$ is called the {\em metric} of the
noncommutative space.
\end{definition}

Let
\[
E_i=\partial_i X,  \quad \tilde E^i=(E_j)^t\ast g^{j i}, \quad
E^i=g^{i j}\ast E_j,
\]
for $i=1, 2, \dots, n$, where $(E_i)^t= \begin{pmatrix}\partial_i X^1 \\
\partial_i X^2
\\ \vdots \\ \partial_i X^m\end{pmatrix}$ denotes the transpose of $E_i$. Define $e\in\bM_m(\cA)$ by
\begin{eqnarray}\label{idempotent}
\begin{aligned}
e:&= \tilde{E}^j\ast E_j \\
&=\begin{pmatrix}
\partial_i X^1\ast g^{i j}\ast \partial_j X^1  & \partial_i X^1\ast g^{i j}\ast \partial_j X^2
& \dots &\partial_i X^1\ast g^{i j}\ast \partial_j X^m \\
\partial_i X^2 \ast g^{i j}\ast  \partial_j X^1  & \partial_i X^2  \ast g^{i j}\ast  \partial_j X^2
& \dots &\partial_i X^2  \ast g^{i j}\ast  \partial_j X^m \\
\dots & \dots & \dots & \dots \\
\partial_i X^m  \ast g^{i j}\ast  \partial_j X^1  & \partial_i X^m  \ast g^{i j}\ast  \partial_j X^2
& \dots &\partial_i X^m  \ast g^{i j}\ast  \partial_j X^m
\end{pmatrix}.
\end{aligned}
\end{eqnarray}
We have the following results.
\begin{proposition}
\begin{enumerate}
\item  Under matrix multiplication, $E_i \ast \tilde{E}^j =\delta^j_i$ for all $i$ and $j$.
\item The $m\times m$ matrix $e$ satisfies $e \ast  e=e$, that is, it is an
idempotent in $\bM_m(\cA)$.
\item The left and right projective $\cA$-modules $\cM=\lAm \ast  e$
and $\tilde \cM=e \ast \rAm$ are respectively spanned by $E_i$ and
$\tilde{E}^i$. More precisely, we have
\[
\cM=\{a^i \ast  E_i \mid a^i\in\cA\}, \quad  \tilde\cM=\{\tilde{E}^i
 \ast  b_i \mid b_i \in\cA\}.
\]
\end{enumerate}
\end{proposition}
\begin{proof}
Note that $g_{i j} = E_i \ast  (E_j)^t$. Thus $E_i \ast  \tilde E^j
= E_i \ast  (E_k)^t \ast  g^{k j} = \delta_i^j$. It then immediately
follows that
\[
e \ast e = \tilde{E^i} \ast  \left(E_i  \ast \tilde{E^j}\right) \ast
E_j = \tilde{E^i}  \ast \delta^j_i  \ast E_j = e.
\]
Obviously $\cM\subset\{a^i \ast  E_i \mid a^i\in\cA\}$ and $
\tilde\cM\subset\{\tilde{E}^i \ast  b_i\mid b_i \in\cA\}$. By the
first part of the proposition, we have
\[
\begin{aligned}
a^i \ast  E_i \ast  e = a^i \ast  \left(E_i \ast  \tilde{E}^j\right)
 \ast E_j = a^j \ast  E_j,
\\
e \ast \tilde{E}^j \ast  b_j= \tilde{E}^i
 \ast \left(E_i \ast \tilde{E}^j\right) \ast  b_j = \tilde{E}^i \ast  b_i.
\end{aligned}
\]
This proves the last claim of the proposition.
\end{proof}
It is also useful to observe that $\tilde\cM=\{(E_i)^t \ast  b_i\mid
b_i \in\cA\}$ since $\g$ is invertible.

We shall denote $\cM$ and $\tilde\cM$ respectively by $TX$ and
$\tilde TX$, and refer to them as the {\em left} and {\em right
tangent bundles} of the noncommutative space $X$. Note that the
definition of the tangent bundles coincides with that in
\cite{CTZZ}.

The proposition below in particular shows that the fibre metric  $\g: TX\otimes_{\R[[\barh]]} \tilde TX
\longrightarrow \cA$ defined in Definition \ref{fibre-metric} agrees with the
metric of the embedded noncommutative space defined in Definition \ref{def:nc-space}.

\begin{proposition}
For any $\zeta=a^i \ast  E_i \in TX$ and $\xi= (E_j)^t \ast b^j\in
\tilde TX$ with $a_i, b_j\in \cA$,
\[
\g: \zeta \otimes \xi \mapsto \g(\zeta, \xi)=a^i  \ast  g_{i j} \ast
b^j.
\]
In particular, $\g(E_i, (E_j)^t) = g_{i j}$.
\end{proposition}
\begin{proof}
Recall from Definition \ref{fibre-metric} that $\g$ is defined by
\eqref{form} with $\Lambda$ being the identity matrix. Thus for any
$\zeta=a^i \ast  E_i \in TX$ and $\xi= (E_j)^t \ast b^j\in \tilde
TX$ with $a_i, b_j\in \cA$,
\[\g(\zeta, \xi)= a^i \ast  E_i  \ast (E_j)^t  \ast b^j = a^i  \ast g_{i j} \ast b^j. \]
This completes the proof.
\end{proof}

Let us now equip the left and right tangent bundles with the {\em
canonical connections} given by $\omega_i= -\tilde\omega_i
=-\partial_i e$, and denote the corresponding covariant derivatives
by
\[
\nabla_i: T X\longrightarrow T X, \quad \tilde\nabla_i: \tilde T
X\longrightarrow \tilde T X.
\]
In principle, one can take arbitrary connections for the tangent
bundles, but we shall not allow this option in this paper.

The following elements of $\cA$ are defined in \cite{CTZZ},
\begin{eqnarray*}
_c\Gamma_{i j l} = \frac{1}{2} \left(\partial_i g_{j l} +
\partial_j g_{l i}
-\partial_l g_{j i} \right), &\quad& \Upsilon_{i j l} = \frac{1}{2}
\left(\partial_i(E_j) \ast  (E_l)^t - E_l \ast \partial_i (E_j)^t\right),\\
\Gamma_{i j l} =  {}_c\Gamma_{i j l} + \Upsilon_{i j
 l}, &\quad& \tilde\Gamma_{i j l} =
{}_c\Gamma_{i j l}- \Upsilon_{i j l},
\end{eqnarray*}
where $\Upsilon_{i j k}$ was referred to as the noncommutative
torsion.  Set \cite{CTZZ}
\begin{eqnarray}\label{Christoffel}
\Gamma_{i j}^k =  \Gamma_{i j l} \ast  g^{l k}, \quad
\tilde\Gamma_{i j}^k = g^{k l} \ast  \tilde\Gamma_{i j l}.
\end{eqnarray}
Then we have the following result.
\begin{lemma}\label{Gamma}
\begin{eqnarray}
\nabla_i E_j= \Gamma_{i j}^k  \ast  E_k, && \tilde \nabla_i \tilde
E^j= - \tilde E^k\ast \Gamma_{k i}^j.
\end{eqnarray}
\end{lemma}
\begin{proof}
Consider the first formula. Write $\partial_i e =
\partial_i(\tilde{E}^k) \ast E_k + \tilde{E}^k \ast \partial_iE_k$. We have
\[
\begin{aligned}
\nabla_i E_j &= \partial_i E_j - E_j \partial_i \ast  e \\
&= \partial_i E_j - \left(\partial_i(E_j  \ast e) - \partial_i(E_j)
 \ast e
\right) \\
&=  \partial_i(E_j) \ast \tilde{E}^k  \ast E_k.
\end{aligned}
\]
It was shown in \cite{CTZZ} that $\Gamma_{i j}^k =\partial_i(E_j)
\ast \tilde{E}^k$. This immediately leads to the first formula. The
proof for the second formula is essentially the same.
\end{proof}

Note that the Lemma \ref{Gamma} can be re-stated as
\[
\nabla_i E^j = - \tilde\Gamma_{i k}^j\ast E^k, \quad \tilde \nabla_i
(E_j)^t= (E_k)^t \ast \tilde \Gamma_{i j}^k.
\]

By using Lemma \ref{metric} and Lemma \ref{Gamma}, we can easily
prove the following result, which is equivalent to \cite[Proposition
2.7]{CTZZ}.
\begin{proposition}\label{prop:compatible}
The connections are metric compatible in the sense that
\begin{eqnarray}\label{compatible}
\partial_i \g(\zeta, \xi)=\g(\nabla_i\zeta, \xi) + \g(\zeta,
\tilde\nabla_i\xi), &\quad& \forall \zeta\in TX, \ \xi\in \tilde TX.
\end{eqnarray}
\end{proposition}
For $\zeta=E_j$ and $\xi=(E_k)^t$, we obtain from \eqref{compatible}
the following result for all $i, j, k$:
\begin{eqnarray}\label{compatible1}
\partial_i g_{j k} -\Gamma_{i j k} - \tilde\Gamma_{i k j} =0.
\end{eqnarray}
This formula is in fact equivalent to Proposition
\ref{prop:compatible}.

Define
\begin{eqnarray}\label{R-tensor}
R^l_{k i j} = E_k \ast \cR_{i j} \ast \tilde E^l, \quad \tilde
R^l_{k i j} =- g^{l q} \ast  E_q \ast \cR_{i j} \ast \tilde E^p \ast
g_{p k}.
\end{eqnarray}
We can show by some lengthy calculations that
\begin{eqnarray}\label{curvature1}
\begin{aligned} R_{k i j}^l &=& -\partial_j\Gamma_{i k}^l -
\Gamma_{i k}^p  \ast  \Gamma_{j p}^l + \partial_i\Gamma_{j k}^l
+\Gamma_{j k}^p  \ast \Gamma_{i p}^l ,\\
\tilde R_{k i j}^l &=& -\partial_j\tilde\Gamma_{i k}^l -
\tilde\Gamma_{j p}^l \ast \tilde\Gamma_{i k}^p
+\partial_i\tilde\Gamma_{j k}^l + \tilde\Gamma_{i p}^l
 \ast \tilde\Gamma_{j k}^p,
\end{aligned}
\end{eqnarray}
which are the {\em Riemannian curvatures} of the left and right
tangent bundles of the noncommutative space $X$ given in \cite[Lemma
2.12 and \S 4]{CTZZ}. Therefore,
\begin{eqnarray}\label{curvature}
{[\nabla_i, \nabla_j]} E_k = R_{k i j}^l  \ast E_l, & \quad &
{[}\tilde\nabla_i, \tilde\nabla_j{]}(E_k)^t = (E_l)^t  \ast \tilde
R_{k i j}^l,
\end{eqnarray}
which were proved in \cite{CTZZ}.
Let $R_{l k i j} = R_{k i j}^p\ast g_{p l}$ and $\tilde R_{l k i j}
= - g_{k p}\ast \tilde R_{l i j}^p$. By \eqref{R:leqr}, $R_{k l i
j}=\tilde R_{ k l i j}$.

\begin{definition}
Let
\begin{eqnarray}
R_{i j} = R^p_{i p j}, &\quad& R= g^{j i}\ast R_{i j},
\end{eqnarray}
and call them the {\em Ricci curvature} and {\em scalar curvature}
of the noncommutative space respectively.
\end{definition}
Then obviously
\begin{eqnarray}\label{Ricci}
R_{i j} = -\g([\nabla_j, \nabla_l]E_i, \tilde E^l), \quad R=
-\g([\nabla_i, \nabla_k] E^i, \tilde E^k).
\end{eqnarray}
\subsection{Second fundamental form}

In the theory of classical surfaces, the second fundamental form
plays an important role. A similar notion exists for embedded noncommutative
spaces.
\begin{definition}\label{sff}
We define the left and right {\em second fundamental forms} of the
noncommutative surface $X$ by
\begin{eqnarray}\label{h}
h_{i j} = \partial_i E_j -\Gamma_{i j}^k \ast E_k, && \tilde h_{i
j}= \partial_i E_j - E_k \ast \tilde \Gamma_{i j}^k.
\end{eqnarray}
\end{definition}
It follows from equation \eqref{Gamma} that
\begin{eqnarray}\label{orth}
h_{i j}\bullet E_k =0, && E_k\bullet \tilde h_{i j} =0.
\end{eqnarray}

The Riemann curvature $R_{l k i j }= (\nabla_i \nabla_j- \nabla_j
\nabla_i) E_k\bullet E_l$ can be expressed in terms of the second
fundamental forms. Note that
\begin{eqnarray*}
R_{l k i j }&=&  \partial_j  E_k\bullet \partial_i E_l -
\partial_j E_k\bullet \tilde \nabla_i E_l - \partial_i E_k\bullet
\partial_j E_l +
\partial_i E_k\bullet \tilde \nabla_j E_l.
\end{eqnarray*}
By Definition \ref{sff},
\begin{eqnarray*}
R_{l k i j }&=& \partial_j E_k\bullet \tilde h_{i l} -
\partial_i E_k\bullet \tilde h_{j l}\\
&=& (\nabla_j E_k +h_{j k})\bullet \tilde h_{i l} - (\nabla_i E_k +
h_{i k})\bullet \tilde h_{j l}.
\end{eqnarray*}
Equation \eqref{orth} immediately leads to the following result.
\begin{lemma} The following generalized Gauss
equation holds:
\begin{eqnarray}
R_{l k i j } &=& h_{j k}\bullet \tilde h_{i l} - h_{i k}\bullet
\tilde h_{j l}.
\end{eqnarray}
\end{lemma}

\subsection{Minkowski signature}\label{Mink}

Now let us briefly comment on noncommutative spaces with Minkowski
signatures embedded in higher dimensions \cite{CTZZ}. Fix a
diagonal $m\times m$
matrix $\eta=diag(-1, \dots, -1, 1, \dots, 1)$
with $p$ of the diagonal entries being $-1$, and $q=m-p$ of
them being $1$. Given $X=\begin{pmatrix} X^1 & X^2 &\dots & X^m
\end{pmatrix}$ in $\lAm$, we define an $n\times n$ matrix $\g=(g_{i
j})_{i, j=1, 2, \dots, n}$ with entries
\[
g_{i j}=\sum_{\alpha=1}^m \partial_i X^\alpha  \ast \eta_{\alpha
\beta}  \ast \partial_j X^\beta.
\]
We call $X$ a {\em noncommutative space} embedded in $\cA^m$ if the
matrix $\g$ is invertible. Denote its inverse matrix by $(g^{i j})$.
Now the idempotent which gives rise to the left and right tangent
bundles of $X$ is given by
\[
e=  \eta (E_i)^t  \ast g^{i j} \ast  E_j,
\]
which obviously satisfies $E_i \ast  e = E_i$ for all $i$.  The
fibre metric of Definition \ref{metric} yields a metric on the
embedded noncommutative surface $X$.

\section{Elementary examples}\label{sect:examples} 

We present several simple examples of noncommutative embedded spaces
to illustrate the general theory developed in earlier sections.
We shall mainly discuss the metric, (canonical) connection and Riemannian
curvature for each noncommutative embedded space. However, in Section \ref{sect:slice},
we examine the tangent bundle as a projective module over the noncommutative algebra of
functions by explicitly constructing the corresponding idempotent.

\subsection{Noncommutative sphere} \label{sphere}

Let $U=(0, \pi)\times(0, 2\pi)$, and we write $\theta$ and $\phi$
for $t_1$ and $t_2$ respectively. Let $X(\theta, \phi)=(X^1(\theta,
\phi), X^2(\theta, \phi), X^3(\theta, \phi))$ be given by
\begin{eqnarray}
X(\theta, \phi) &=& \left(\frac{\sin\theta\cos\phi}{\cosh\barh},
\frac{\sin\theta\sin\phi}{\cosh\barh}, \frac{\sqrt{\cosh 2\barh }
\cos\theta}{\cosh\barh}\right)
\end{eqnarray}
with the components being smooth functions in $(\theta, \phi)\in U$.
It can be shown that $X$ satisfies the following relation
\begin{eqnarray}\label{eq:sphere}
X^1\ast X^1 + X^2\ast X^2 + X^3\ast X^3 =1.
\end{eqnarray}
Thus we may regard the noncommutative surface defined by $X$ as an
analogue of the sphere $S^2$. We shall denote it by $S^2_\barh$ and
refer to it as a {\em noncommutative sphere}. We have
\begin{eqnarray*}
E_1 &=&\left(\frac{\cos\theta\cos\phi}{\cosh\barh},
\frac{\cos\theta\sin\phi}{\cosh\barh}, -\frac{\sqrt{\cosh 2\barh }
\sin\theta}{\cosh\barh}\right),\\
E_2&=&\left(-\frac{\sin\theta\sin\phi}{\cosh\barh},
\frac{\sin\theta\cos\phi}{\cosh\barh}, 0\right).
\end{eqnarray*}
The components  $g _{ij}=E _i \bullet E _j$  of the metric $\g$ on
$S^2_\barh$ can now be calculated, and we obtain
\begin{eqnarray*}
\begin{aligned}
&g_{1 1} = 1 ,  \quad g_{2 2}= \sin^2\theta -
\frac{\sinh^2\barh}{\cosh^2\barh}\cos^2\theta  ,\\
&g_{1 2} =- g_{2 1}= \frac{\sinh\barh}{\cosh\barh}\left(\sin^2\theta
- \cos^2\theta\right).
\end{aligned}
\end{eqnarray*}

The components of this metric commute with one another as they
depend on $\theta$ only. Thus it makes sense to consider the usual
determinant $G$ of $\g$. We have
\begin{eqnarray*}
\begin{aligned}
G =&\sin ^2 \theta  + \tanh ^2 \barh (\cos ^2 2 \theta -\cos
^2 \theta) \\
=&\sin ^2 \theta [1 + \tanh ^2 \barh (1-4\cos ^2  \theta)].
\end{aligned}
\end{eqnarray*}
The inverse metric is given by
\begin{eqnarray*}
\begin{aligned}
&g^{1 1} = \frac{\sin ^2 \theta -\tanh^2\barh \cos ^2 \theta}
{\sin ^2 \theta +\tanh^2\barh (\cos ^2 2\theta -\cos ^2 \theta)} ,  \\
&g^{2 2}= \frac{1}{\sin ^2 \theta +\tanh^2\barh (\cos ^2 2\theta -\cos ^2 \theta)} ,\\
&g^{1 2} =- g^{2 1}=\frac{\tanh \barh \cos 2 \theta} {\sin ^2 \theta
+\tanh^2\barh (\cos ^2 2\theta -\cos ^2 \theta)}.
\end{aligned}
\end{eqnarray*}

Now we determine the connection and curvature tensor of the
noncommutative sphere. The computations are quite lengthy, thus we
only record the results here. For the Christoffel symbols, we have
\begin{eqnarray*}
\begin{aligned}
& \Gamma _{1 1 1} = \tilde\Gamma _{1 1 1}=0, &\quad& \Gamma _{1 1 2}
= -\tilde\Gamma _{1 1 2} =
\sin 2\theta \tanh \barh,  \\
& \Gamma _{1 2 1} = -\tilde\Gamma _{1 2 1}=-\sin 2\theta \tanh
\barh,
&\quad& \Gamma _{1 2 2} = \tilde\Gamma _{1 2 2} =\frac{1}{2}\sin 2\theta (1+\tanh ^2 \barh),  \\
& \Gamma _{2 1 1} = -\tilde\Gamma _{2 1 1}=\sin 2\theta \tanh \barh,
&\quad& \Gamma _{2 1 2} = \tilde\Gamma _{2 1 2} =\frac{1}{2}\sin
2\theta (1+\tanh ^2 \barh),\\
&\Gamma _{2 2 1} = \tilde\Gamma _{2 2 1} =-\frac{1}{2}\sin 2\theta
(1+\tanh ^2 \barh), &\quad& \Gamma _{2 2 2} = -\tilde\Gamma _{2 2 2}
=\sin 2\theta \tanh \barh.
\end{aligned}
\end{eqnarray*}
Note that $\Gamma _{1 1 2}\ne \tilde\Gamma _{1 1 2}$.
We now find the asymptotic expansions of the
curvature tensors with respect to $\barh$:
\begin{eqnarray*}
\begin{aligned}
R _{1 1 1 2} = &2 \barh +(\frac{10}{3} +4 \cos 2 \theta)\barh ^3+O(\barh ^4),  \\
 R _{2 1 1 2} = &-\sin ^2 \theta -\frac{1}{2}(4+\cos 2 \theta -\cos 4
 \theta)
 \barh ^2+O(\barh
 ^4),\\
R _{1 2 1 2} = &\sin ^2 \theta +\frac{1}{2}(4+\cos 2 \theta -\cos 4
\theta) \barh ^2+O(\barh
 ^4),\\
R _{2 2 1 2} = &-2\sin ^2 \theta \barh -(\frac{5}{3}+\frac{4}{3}\cos
2 \theta -4 \cos 4 \theta) \barh ^3+O(\barh ^4).
\end{aligned}
\end{eqnarray*}
We can also compute asymptotic expansions of the Ricci curvature
tensor
\begin{eqnarray*}
\begin{aligned}
 R _{1 1} = &1+(6+4\cos 2\theta)\barh ^2 +O(\barh ^4),  \\
 R _{2 1} = &(2-\cos 2\theta)\barh +\frac{1}{3}
 (16+19\cos 2 \theta -6 \cos 4\theta )\barh ^3 +O(\barh ^4),  \\
 R _{1 2} = &(2+\cos 2\theta)\barh +\frac{1}{3}
 (16+29\cos 2 \theta +6 \cos 4\theta )\barh ^3 +O(\barh ^4),  \\
 R _{2 2} = &\sin ^2 \theta +\frac{1}{2}
 (3+5\cos 2 \theta -2\cos 4 \theta)\barh ^2+O(\barh ^4),
\end{aligned}
\end{eqnarray*}
and the scalar curvature
\begin{eqnarray*}
 R  = 2+4(3+4\cos 2 \theta) \barh ^2+O(\barh
 ^4).
\end{eqnarray*}

By setting $\barh=0$, we obtain from the various curvatures of
$S^2_\barh$ the corresponding objects for the usual sphere $S^2$.
This is a useful check that our computations above are accurate.

\subsection{Noncommutative torus} \label{torus}

This time we shall take $U=(0, 2\pi)\times(0, 2\pi)$, and denote a
point in $U$ by $(\theta, \phi)$. Let $X(\theta, \phi)=(X^1(\theta,
\phi), X^2(\theta, \phi), X^3(\theta, \phi))$ be given by
\begin{eqnarray}
X(\theta, \phi) &=& \left((a+\sin\theta)\cos\phi,
(a+\sin\theta)\sin\phi,  \cos\theta\right)
\end{eqnarray}
where $a>1$ is a constant. Classically $X$ is the torus. When we
extend scalars from $\R$ to $\R[[\barh]]$ and impose the star
product on the algebra of smooth functions, $X$ gives rise to a
noncommutative torus, which will be denoted by $T ^2 _\barh$. We
have
\begin{eqnarray*}
E_1 &=&\left(\cos\theta\cos\phi, \cos\theta\sin\phi,
-\sin\theta\right),\\
E_2&=&\left(-(a+\sin\theta)\sin\phi, (a+\sin\theta)\cos\phi,
0\right).
\end{eqnarray*}
The components  $g _{ij}=E _i \bullet E _j$  of the metric $\g$ on
$T^2_\barh$ take the form
\begin{eqnarray*}
\begin{aligned}
&g_{1 1} = 1 +\sinh ^2 \barh \cos 2\theta,  \\
&g_{2 2}= (a+\cosh\barh \sin\theta) ^2 -\sinh^2\barh \cos^2\theta  ,\\
&g_{1 2} =- g_{2 1}=- \sinh\barh \cosh\barh \cos 2\theta +a \sinh
\barh \sin \theta.
\end{aligned}
\end{eqnarray*}
As they depend only on $\theta$, the components of the metric
commute with one another. The inverse metric is given by
\begin{eqnarray*}
\begin{aligned}
&g^{1 1} = \frac{(a+\cosh\barh \sin\theta) ^2 -\sinh^2\barh
\cos^2\theta}{G},  \\
&g^{2 2}= \frac{1+\sinh ^2 \barh \cos 2\theta}{G} ,\\
&g^{1 2} =- g^{2 1}=\frac{\sinh\barh \cosh\barh \cos 2\theta +a
\sinh \barh \sin \theta}{G},
\end{aligned}
\end{eqnarray*}
where $G$ is the usual determinant of $\g$ given by
\begin{eqnarray*}
G =(\sin \theta  + a\cosh \barh ) ^2-a ^2 \sin ^2 \theta \sinh ^2
\barh.
\end{eqnarray*}

Now we determine the curvature tensor of the noncommutative torus.
The computations can be carried out in much the same way as in the
case of the noncommutative sphere, and we merely record the results
here. For the connection, we have
\begin{eqnarray*}
\begin{aligned}
& \Gamma _{1 1 1} = - \sin 2\theta \sinh^2\barh, &&
\Gamma_{1 1 2} = (a \cos\theta  + \sin2\theta \cosh\barh)\sinh\barh,  \\
& \Gamma _{1 2 1} = -\sin 2\theta  \sinh \barh \cosh \barh,
&& \Gamma _{1 2 2} = a \cos \theta \cosh \barh+\frac{1}{2}\sin 2\theta \cosh 2\barh,  \\
& \Gamma _{2 1 1} = -\sin 2\theta \sinh \barh \cosh \barh, &&
\Gamma _{2 1 2} = a \cos \theta \cosh \barh +\frac{1}{2}\sin
2\theta \cosh 2 \barh,\\
&\Gamma _{2 2 1} = -(a \cos \theta  +\frac{1}{2}\sin
2\theta) \cosh \barh, && \Gamma _{2 2 2} = (2a \cos\theta  + \sin 2\theta \cosh \barh)\sinh \barh.
\end{aligned}
\end{eqnarray*}
We can find the asymptotic expansions of the curvature tensors with
respect to $\barh$:
\begin{eqnarray*}
\begin{aligned}
 R _{1 1 1 2} = &\frac{2\sin \theta (1+a \sin \theta)}{a+\sin \theta} \barh +O(\barh ^3),  \\
 R _{2 1 1 2} = &-\sin \theta (a+\sin \theta)+O(\barh ^2),\\
 R _{1 2 1 2} = &\sin \theta (a+\sin \theta)+O(\barh ^2),\\
 R _{2 2 1 2} = &-2\sin ^2 \theta (1+a \sin \theta) \barh +O(\barh ^3).
\end{aligned}
\end{eqnarray*}
We can also compute asymptotic expansions of the Ricci curvature
tensor
\begin{eqnarray*}
\begin{aligned}
 R _{1 1} = &\frac{\sin \theta}{a+\sin \theta} +O(\barh ^2),  \\
 R _{2 1} = &-\frac{\sin \theta(-3a +5a \cos \theta -(5+2a ^2)
 \sin \theta +\sin 3\theta)}{2(a+\sin \theta) ^2}\barh +O(\barh ^3),  \\
 R _{1 2} = &\frac{\sin \theta(a+\cos 2 \theta +a \sin \theta)}{a+\sin \theta}\barh +O(\barh ^3),  \\
 R _{2 2} = &\sin \theta (a+\sin \theta)+O(\barh ^2),
\end{aligned}
\end{eqnarray*}
and the scalar curvature
\begin{eqnarray*}
R  = \frac{2\sin \theta}{a+\sin \theta}+O(\barh^2).
\end{eqnarray*}
By setting $\barh=0$, we obtain from the various curvatures of
$T^2_\barh$ the corresponding objects for the usual torus $T^2$.

\subsection{Noncommutative hyperboloid} \label{hyperbolid}

Another simple example is the noncommutative analogue of the
hyperboloid described by $X= (x, y, {\sqrt{1+x^2+y^2}})$. One may
also change the parametrization and consider instead
\begin{eqnarray}
X(r, \phi) &=& \left(\sinh r \cos \phi, \sinh r \sin \phi, \cosh
r\right)
\end{eqnarray}
on $U=(0, \infty)\times(0, 2\pi)$, where a point in $U$ is denoted
by $(r, \phi)$. When we extend scalars from $\R$ to $\R[[\barh]]$
and impose the star product on the algebra of smooth functions
(of $t_1=r$ and $t_2=\phi$), $X$ gives
rise to a noncommutative hyperboloid, which will be denoted by $H ^2
_\barh$. We have
\begin{eqnarray*}
E_1 &=&\left(\cosh r \cos\phi, \cosh r \sin\phi, \sinh r\right),\\
E_2&=&\left(-\sinh r \sin\phi, \sinh r \cos\phi, 0\right).
\end{eqnarray*}
The components  $g _{ij}=E _i \bullet E _j$  of the metric $\g$ on
$H^2_\barh$ take the form
\begin{eqnarray*}
\begin{aligned}
&g_{1 1} = \cos ^2 \barh \cosh 2r,  \\
&g_{2 2}= \frac{1}{2} \left(-1+\cos 2\barh \cosh 2r\right),\\
&g_{1 2} =- g_{2 1}=- \frac{1}{2}\sin 2\barh \cosh 2r.
\end{aligned}
\end{eqnarray*}
As they depend only on $r$, the components of the metric commute
with one another. The inverse metric is given by
\begin{eqnarray*}
\begin{aligned}
&g^{1 1} =\frac{\sec^2 \barh}{2\sinh^2 r }\left(\cos 2\barh-\frac{1}{\cosh 2r} \right),\\
&g^{2 2}= \frac{1}{\sinh^2 r},\\
&g^{1 2} =- g^{2 1}=\frac{\tan\barh}{\sinh^2 r}.
\end{aligned}
\end{eqnarray*}

Now we determine the curvature tensor of the noncommutative
hyperboloid. For the connection, we have
\begin{eqnarray*}
\begin{aligned}
& \Gamma _{1 1 1} = \cos^2 \barh \sinh 2r, &\quad& \Gamma _{1 1 2} =
-\frac{1}{2}\sin 2\barh \sinh 2r,\\
& \Gamma _{1 2 1} = \frac{1}{2}\sin 2\barh \sinh 2r, &\quad& \Gamma
_{1 2 2} = \frac{1}{2}\cos 2\barh \sinh 2r,  \\
& \Gamma _{2 1 1} =\frac{1}{2}\sin 2\barh \sinh 2r, &\quad& \Gamma
_{2 1 2} = \frac{1}{2}\cos 2\barh \sinh 2r,\\
&\Gamma _{2 2 1} = -\frac{1}{2}\cos 2\barh \sinh 2r, &\quad& \Gamma
_{2 2 2} =\frac{1}{2}\sin 2\barh \sinh 2r.
\end{aligned}
\end{eqnarray*}
We can find the asymptotic expansions of the curvature tensors with
respect to $\barh$:
\begin{eqnarray*}
\begin{aligned}
 R _{1 1 1 2} = &\frac{2}{\cosh 2r}\barh +O(\barh ^2),\\
 R _{2 1 1 2} = &-\frac{\sinh ^2 r}{\cosh 2r} +O(\barh ^2),\\
 R _{1 2 1 2} = &\frac{\sinh ^2 r}{\cosh 2r} +O(\barh ^2),\\
 R _{2 2 1 2} = &-\frac{\cosh 2r +\sinh ^2 2r}{\cosh 2r}\barh +O(\barh ^3).
\end{aligned}
\end{eqnarray*}
We can also compute asymptotic expansions of the Ricci curvature
tensor
\begin{eqnarray*}
\begin{aligned}
 R _{1 1} = &\frac{1}{\cosh 2r} +O(\barh ^2),  \\
 R _{2 1} = &\frac{\coth^2 r (2\cosh 2r -1)}{\cosh 2r} \barh +O(\barh ^3),  \\
 R _{1 2} = &\frac{\cosh 2r +2}{\cosh^2 2r}\barh +O(\barh ^3),  \\
 R _{2 2} = &\frac{\sinh ^2 r}{\cosh ^2 2r}+O(\barh ^2),
\end{aligned}
\end{eqnarray*}
and the scalar curvature
\begin{eqnarray*}
R  = \frac{2}{\cosh ^2 2r}+O(\barh^2).
\end{eqnarray*}
By setting $\barh=0$, we obtain from the various curvatures of
$H^2_\barh$ the corresponding objects for the usual hyperboloid
$H^2$.

\subsection{A time slice of a quantised Schwarzschild spactime}\label{sect:slice}
We analyze an embedded noncommutative surface of Euclidean signature
arising from the quantisation of a time slice of the Schwarzschild
spacetime. While the main purpose here is to illustrate how the
general theory developed in previous sections works, the example is
interesting in its own right.

Let us first specify the notation to be used in this section. Let
$t^1=r$, \ $t^2=\theta$ and $t^3=\phi$, with $r>2m$, $\theta\in (0,
\pi)$, and $\phi\in (0, 2\pi)$. We deform the algebra of functions
in these variables by imposing the Moyal product defined by
\eqref{multiplication} with the following anti-symmetric matrix
\begin{eqnarray*}
\left(\theta _{i j}\right)_{i, j=1}^3=\left(\begin{array}{cccc}
    0&  0&  0\\
    0&  0&  1\\
    0&  -1&  0
\end{array}\right).
\end{eqnarray*}
Note that the functions depending  only on the variable $r$ are
central in the Moyal algebra $\cA$. We shall write the usual
pointwise product of two functions $f$ and $\g$ as $f g$, but write
their Moyal product as $f\ast g$.

Consider $X=\begin{pmatrix} X^1 & X^2 & X^3 & X^4\end{pmatrix}$ with
\begin{eqnarray}
\begin{aligned}
X^1=f(r) \quad \text{with} \quad (f')^2 +1=\left(1-\frac{2m}{r}\right)^{-1},\\
X^2=r \sin\theta \cos\phi,\quad X^3=r \sin\theta \sin\phi, \quad
X^4=r \cos\theta.
\end{aligned}
\end{eqnarray}
Simple calculations yield
\[
\begin{aligned}
E_1&=\partial_r X=\begin{pmatrix}f' &\sin\theta \cos\phi &\sin\theta \sin\phi &\cos\theta\end{pmatrix}, \\
E_2&=\partial_\theta X=\begin{pmatrix}0 &r\cos\theta \cos\phi
&r\cos\theta \sin\phi
&-r\sin\theta\end{pmatrix},\\
E_3&=\partial_\phi X=\begin{pmatrix}0 &-r\sin\theta \sin\phi
&r\sin\theta \cos\phi &0\end{pmatrix}.
\end{aligned}
\]
Using these formulae, we obtain the following expressions for the
components of the metric of the noncommutative surface $X$:
\begin{eqnarray}
\begin{aligned}
g_{1 1}=&\left(1-\frac{2m}{r}\right)^{-1}
\left[1-\left(1-\frac{2m}{r}\right)
          \cos(2\theta)\sinh^2\barh\right],\\
g_{1 2} =    &g_{2 1} = r\sin(2\theta)\sinh^2\barh,\\
g_{2 2} =    &r^2
    \left[1+\cos(2\theta)\sinh^2\barh\right],\\
g_{2 3} = &-g_{3 2}
=-r^2\cos(2\theta)\sinh\barh\cosh\barh,\\
g_{1 3} = &-g_{3 1}
=-r\sin(2\theta)\sinh\barh\cosh\barh,\\
g_{3 3} =&r^2
    \left[\sin^2\theta - \cos(2\theta)\sinh^2\barh\right].
\end{aligned}
\end{eqnarray}
In the limit $\barh\rightarrow 0$, we recover the spatial components
of the Schwarzschild metric.  Observe that the noncommutative
surface still reflects the characteristics of the Schwarzschild
spacetime in that there is a time slice of the Schwarzschild black
hole with the event horizon at $r = 2m$.

Since the metric $\g$ depends on $\theta$ and $r$ only, and the two
variables commute, the inverse $(g^{i j})$ of the metric can be
calculated in the usual way as in the commutative case. Now the
components of the idempotent $e=(e_{i j})= (E_i)^t * g^{i j} * E_j$
are given by the following formulae:
\begin{eqnarray*}
\begin{aligned}
e_{11} =&\frac{2m}{r} + \frac{2m(2m - r)(2 + \cos 2\theta)}{r^2} \barh^2+ O(\barh^3),\\
e_{12} =&\frac{m \cos\phi \sin\theta}{r \sqrt{\frac{m}{-4m + 2r}}}
        - \frac{2m\cos\theta\sin\phi }{r\sqrt{\frac{m}{-4m + 2r}}}
        \barh\\
        &+
  \frac{m( 4m + r + 2m\cos 2\theta ) \cos\phi \sin\theta }{r^2 \sqrt{\frac{m}{-4m + 2r}}} \barh^2+ O(\barh^3)\\
e_{13}=&\frac{m\sin\theta\sin\phi}{r\sqrt{\frac{m}{-4m + 2r}}} +
\frac{2m\cos\theta\cos\phi }{r{\sqrt{\frac{m}{-4m + 2r}}}} \barh \\
&+ \frac{m( 4m + r + 2m\cos 2\theta) \sin\theta \sin\phi }{r^2 {\sqrt{\frac{m}{-4m + 2r}}}} \barh^2 + O(\barh^3)\\
e_{14}=&\frac{m\cos\theta}{r{\sqrt{\frac{m}{-4m + 2r}}}} +
\frac{m\cos\theta( 4m - r + 2m\cos 2\theta)}
   {r^2 {\sqrt{\frac{m}{-4m + 2r}}}} \barh^2 + O(\barh^3)\\
\end{aligned}
\end{eqnarray*}
\begin{eqnarray*}
\begin{aligned}
e_{21}=&\frac{m\cos\phi\sin\theta}{r{\sqrt{\frac{m}{-4m + 2r}}}} +
\frac{2m\cos\theta \sin\phi }{r{\sqrt{\frac{m}{-4\,m +
2\,r}}}}\barh\\ &+\frac{m ( 4m + r + 2m\cos 2\theta) \cos\phi
\sin\theta }{r^2 {\sqrt{\frac{m}{-4m + 2r}}}} \barh^2
+ O(\barh^3)\\
e_{22}=&1-\frac{2 m\sin^2 \theta \cos^2 \phi }{r} \\
&+ \frac{m}{2 r^2}\Big[2r + 2m\cos 4\theta \cos^2 \phi - 6 m\cos^2
\phi \\
&+2\cos 2\theta (m + 8r + ( m - r)\cos 2\phi ) \Big] \barh^2 + O(\barh^3)\\
e_{23}=&-\frac{m \sin^2 \theta \sin 2\phi}{r} - \frac{3m\sin 2\theta }{r} \barh\\
&+ \frac{m( 2( m - r)\cos 2\theta + m( -3 + \cos 4\theta ) ) \sin
2\phi}{2r^2} \barh^2+
  O(\barh^3)\\
e_{24}=&\frac{-2m \cos\theta \cos\phi \sin\theta}{r} - \frac{m( 1 + 3\cos 2\theta)\sin\phi }{r} \barh\\
&-\frac{m ( 8 m + 5r + 4m \cos 2\theta ) \cos\phi \sin 2\theta }{2 r^2} \barh^2+ O(\barh^3)\\
\end{aligned}
\end{eqnarray*}
\begin{eqnarray*}
\begin{aligned}
e_{31}=&\frac{m\sin\theta\sin\phi}{r{\sqrt{\frac{m}{-4m + 2r}}}}
-\frac{2m\cos\theta\cos\phi }{r{\sqrt{\frac{m}{-4m + 2r}}}} \barh\\
&+   \frac{m( 4m + r + 2m\cos 2\theta) \sin\theta \sin\phi }{r^2 {\sqrt{\frac{m}{-4m + 2r}}}} \barh^2 + O(\barh^3)\\
e_{32}=&-\frac{m \sin^2 \theta \sin 2\phi}{r} + \frac{3m\sin 2\theta }{r} \barh \\
&+
  \frac{m( 2( m - r)\cos 2\theta + m( -3 + \cos 4\theta ) ) \sin 2\phi }{2r^2} \barh^2+
  O(\barh^3)\\
e_{33}=&1-\frac{2 m\sin^2 \theta \sin^2 \phi }{r} \\
&+ \frac{m}{2 r^2}\Big[ 2r + 2m\cos 4\theta \sin^2 \phi - 6m\sin^2
\phi \\
&+ 2\cos 2\theta (m + 8r - ( m - r)\cos 2\phi ) \Big] \barh^2 + O(\barh^3)\\
e_{34}=&\frac{-2m\cos\theta \sin\theta \sin\phi }{r} + \frac{m( 1 + 3\cos 2\theta ) \cos\phi }{r} \barh \\
&-
  \frac{m( 8m + 5r + 4m \cos 2\theta) \sin 2\theta \sin\phi }{2 r^2} \barh^2+ O(\barh^3)\\
\end{aligned}
\end{eqnarray*}
\begin{eqnarray*}
\begin{aligned}
e_{41}=&\frac{m\cos\theta}{r{\sqrt{\frac{m}{-4m + 2r}}}} +
\frac{m\cos\theta( 4m - r + 2m\cos 2\theta)}
   {r^2 {\sqrt{\frac{m}{-4m + 2r}}}} \barh^2 + O(\barh^3)\\
e_{42}=&\frac{-2m \cos\theta \cos\phi \sin\theta}{r} + \frac{m( 1 + 3\cos 2\theta)\sin\phi }{r} \barh\\
&-
  \frac{m ( 8 m + 5r + 4m \cos 2\theta ) \cos\phi \sin 2\theta }{2 r^2} \barh^2+ O(\barh^3)\\
e_{43}=&\frac{-2m\cos\theta \sin\theta \sin\phi }{r} -\frac{m( 1 + 3\cos 2\theta ) \cos\phi }{r} \barh \\
&-
  \frac{m( 8m + 5r + 4m \cos 2\theta) \sin 2\theta \sin\phi }{2 r^2} \barh^2 + O(\barh^3)\\
e_{44}=&1-\frac{2m \cos^2\theta}{r} + \frac{4m \cos^2 \theta ( -2m +
r - m\cos 2\theta )}{r^2} \barh^2 +O(\barh^3)
\end{aligned}
\end{eqnarray*}

Here we refrain from presenting the result of the Mathematica
computation for the curvature $\cR_{i j}=-[\partial_i e,
\partial_j e]$, which is very complicated and
not terribly illuminating. A detailed analysis of a quantised
Schwarzschild spacetime will be given in Section \ref{schwarzschild}.

\section{General coordinate transformations}\label{sect:transformations}

We now  return to the general setting of Section \ref{sect:bundles}
to investigate ``general coordinate transformations". Our treatment
follows closely \cite[\S V]{CTZZ} and makes use of general ideas of
\cite{Ge, D, Ko}. We should point out that the material presented is
part of an attempt of ours to develop a notion of ``general
covariance" in the noncommutative setting. This is an important
matter which deserves a thorough investigation. We hope that the
work presented here will prompt further studies.

Let $(\cA, \mu)$ be a Moyal algebra of smooth functions on the open
region $U$ of $\R^n$ with coordinate $t$. This algebra is defined
with respect to a constant skew symmetric matrix $\theta=(\theta_{i
j})$. Let $\Phi: U\longrightarrow U$ be a diffeomorphism of $U$ in
the classical sense. We denote \[ u^i=\Phi^i(t), \] and refer to
this as a {\em general coordinate transformation} of $U$.

Denote by $\cA_u$ the sets of smooth functions of $u=(u^1, u^2,
\dots, u^n)$. The map $\Phi$ induces an $\R[[\barh]]$-module
isomorphism $\phi=\Phi^*: \cA_u\longrightarrow \cA$ defined for any
function $f\in\cA_u$ by
\[\phi(f)(t) =f(\Phi(t)).\]
We define the $\R[[\barh]]$-bilinear map
\[
\mu_u: \cA_u\otimes \cA_u \longrightarrow \cA_u, \quad \mu_u(f,
g)= \phi^{-1} \mu_t(\phi(f), \phi(g)).
\]
Then it is well-known \cite{Ge} that $\mu_u$ is associative.
Therefore, we have the associative algebra isomorphism
\[
\phi: (\cA_u, \mu_u) \stackrel{\sim}{\longrightarrow} (\cA_t,
\mu_t).
\]
We say that the two associative algebras are {\em gauge equivalent}
by adopting the terminology of \cite{D}.

Following \cite{CTZZ}, we define $\R[[\barh]]$-linear operators
\begin{eqnarray}\label{dphi}
\partial_i^\phi := \phi^{-1}\circ \partial_i\circ \phi: \cA_u \longrightarrow
\cA_u,
\end{eqnarray}
which have the following properties \cite[Lemma 5.5]{CTZZ}:
\[
\begin{aligned}
&\partial_i^\phi \circ \partial_j^\phi - \partial_j^\phi \circ
\partial_i^\phi=0, \\
&\partial_i^\phi\mu_u(f, g) = \mu_u(\partial_i^\phi(f),  g) +
\mu_u(f, \partial_i^\phi(g)), \quad \forall f, g\in \cA_u,
\end{aligned}
\]
where the second relation is the Leibniz rule for $\partial_i^\phi$.
Recall that this Leibniz rule played a crucial role in the
construction of noncommutative spaces over $(\cA_u, \mu_u)$ in
\cite{CTZZ}.

We shall denote by $\bM_m(\cA_u)$ the set of $m\times m$-matrices
with entries in $\cA_u$. The product of two such matrices will be
defined with respect to the multiplication $\mu_u$ of the algebra
$(\cA_u, \mu_u)$. Then $\phi^{-1}$ acting component wise gives rise
to an algebra isomorphism from $\bM_m(\cA)$ to $\bM_m(\cA_u)$, where
matrix multiplication in $\bM_m(\cA)$ is defined with respect to
$\mu$.

Since we need to deal with two different algebras $(\cA, \mu)$ and
$(\cA_u, \mu_u)$ simultaneously in this section, we write $\mu$ and
the matrix multiplication defined with respect to it by $\ast$ as
before, and use $\ast_u$ to denote $\mu_u$ and the matrix
multiplication defined with respect to it.

Let $e\in\bM_m(\cA)$ be an idempotent. There exists the
corresponding finitely generated projective left (resp. right)
$\cA$-module $\cM$ (resp. $\tilde\cM$). Now $e_u:=\phi^{-1}(e)$ is
an idempotent in $\bM_m(\cA_u)$, that is,
$\phi^{-1}(e)\ast_u\phi^{-1}(e)=\phi^{-1}(e)$. Write
$e_u=(\cE_\alpha^\beta)_{\alpha, \beta=1, \dots, m}$. This
idempotent gives rises to the left projective $\cA_u$-module $\cM_u$
and right projective $\cA_u$-module $\tilde\cM_u$, respectively
defined by
\begin{eqnarray*}
&\cM_u = \left\{\left.\begin{pmatrix}a^\alpha \ast_u\cE_\alpha^1&
a^\alpha \ast_u\cE_\alpha^2 & \dots & a^\alpha \ast_u\cE_\alpha^m
\end{pmatrix}\right| a^\alpha\in \cA_u\right\},\\
&\tilde\cM_u = \left\{\left.\begin{pmatrix}
\cE_1^\beta \ast_u b_\beta \\ \cE_2^\beta\ast_u b_\beta\\
\vdots \\ \cE_m^\beta \ast_u b_\beta
\end{pmatrix}\right| b^\beta\in \cA_u\right\},
\end{eqnarray*}
where $a^\alpha \ast_u \cE_\alpha^\beta = \sum_{\alpha}
\mu_u(a^\alpha, \cE_\alpha^\beta)$ and $\cE_\alpha^\beta\ast_u
b_\beta =\sum_{\beta} \mu_u(\cE_\alpha^\beta, b_\beta)$. Below we
consider the left projective module only, as the right projective
module may be treated similarly.

Assume that we have the left connection
\[
\nabla_i: \cM\longrightarrow \cM, \quad \nabla_i\zeta
=\frac{\partial\zeta}{\partial t^i} + \zeta\ast\omega_i.
\]
Let $\omega_i^u:=\phi^{-1}(\omega_i)$. We have the following result.
\begin{theorem} \label{covariance}
\begin{enumerate}
\item The matrices $\omega_i^u$ satisfy the following relations in
$\bM_m(\cA_u)$:
\[ e_u\ast_u \omega_i^u \ast_u (1- e_u) = - e_u \ast_u  \partial_i^\phi e_u.\]
\item The operators $\nabla^\phi_i$  ($i=1, 2, \dots, n$) defined for all $\eta\in\cM_u$ by
\[ \nabla^\phi_i\eta =\partial_i^\phi\eta + \eta\ast_u \omega_i^u\]
give rise to a connection on $\cM_u$.
\item The curvature of the connection $\nabla^\phi_i$ is given by
\[\cR_{i j}^u = \partial_i^\phi\omega_j^u- \partial_j^\phi\omega_i^u
-\omega_i^u\ast_u \omega_j^u+ \omega_j^u\ast_u \omega_i^u,\] which
is related to the curvature $\cR_{i j}$ of $\cM$ by
\[\cR_{i j}^u = \phi^{-1}(\cR_{i j}).\]
\end{enumerate}
\end{theorem}
\begin{proof}Note that $e_u \ast_u\omega_i^u \ast_u(1- e_u)
=\phi^{-1}(e\ast\omega_i\ast(1-e))$.
We also have $\partial_i^\phi e_u= \phi^{-1}(\frac{\partial
e}{\partial t^i})$, which leads to $e_u\ast_u \partial_i^\phi
e_u=\phi^{-1}(e\ast \phi(\partial_i^\phi e_u))=
\phi^{-1}(e\ast\partial_i e)$. This proves part (1). Part (2)
follows from part (1) and the Leibniz rule for $\partial_i^\phi$.
Straightforward calculations show that the curvature of the
connection $\nabla^\phi_i$ is given by $\cR_{i j}^u
=\partial_i^\phi\omega_j^u- \partial_j^\phi\omega_i^u
-\omega_i^u\ast_u\omega_j^u +  \omega_j^u\ast_u\omega_i^u$. Now
$\partial_i^\phi\omega_j^u = \phi^{-1}(\frac{\partial
\omega_j}{\partial t^i})$, and
$\omega_i^u\ast_u\omega_j^u-\omega_j^u\ast_u\omega_i^u=\phi^{-1}(\omega_i\ast\omega_j)
-\phi^{-1}(\omega_j\ast\omega_i)$. Hence $\cR_{i j}^u =
\phi^{-1}(\cR_{i j})$.
\end{proof}

\begin{remark}\label{classical}
One can recover the usual transformation rules of tensors
under the diffeomorphism group from the commutative limit of Theorem  \ref{covariance}
in a way similar to that in \cite[\S 5.C]{CTZZ}.
\end{remark}

\section{Noncommutative Einstein field equations and exact solutions}\label{sect:Einstein}
\subsection{Noncommutative Einstein field equations}
Recall that in classical Riemannian geometry, the second Bianchi
identity suggests the correct form of Einstein's equation. Let us
make some analysis of this point here.

In Section \ref{sect:surfaces},  we introduced the Ricci curvature
$R_{i j}$ and scalar curvature $R$. Let
\[
R^i_j = g^{i k}\ast R_{k j},
\]
then the scalar curvature is $R=R^i_i$. Let us also introduce the
following object:
\begin{eqnarray}\label{Theta}
\Theta^l_p :=\g([\nabla_p, \nabla_i]E^i,\,  \tilde E^l)= g^{i k}\ast
R^l_{k p i}.
\end{eqnarray}
In the commutative case, $\Theta^l_p$ coincides with $R^l_p$, but
it is no longer true in the present setting. However, note that
\begin{eqnarray}\label{Scalar2}
\Theta^l_l =g ^{ik} \ast R ^l _{kli} = g ^{ik} \ast R _{ki} =R.
\end{eqnarray}

By first contracting the indices $j$ and $l$ in the second Bianchi
identity, then raising the index $k$ to $i$ by multiplying the
resulting identity by $g^{i k}$ from the left and summing over $i$,
we obtain the identity
\begin{eqnarray*}
\begin{aligned}
0&=\partial_p R & - &\partial_i R^i_p&+& \g\left([\nabla_i,
\nabla_l]\nabla_p E^i, \tilde E^l \right)&+&\g\left([\nabla_l,
\nabla_p]\nabla_i E^i, \tilde E^l \right)\\
& &-&\partial_l \Theta^l_p&+& \g\left([\nabla_i, \nabla_l]E^i,
\tilde \nabla_p \tilde E^l \right)&+& \g\left([\nabla_p,
\nabla_i]E^i, \tilde \nabla_l \tilde E^l \right)\\
& && &+&\g\left([\nabla_p, \nabla_i]\nabla_l E^i, \tilde E^l
\right)&+& \g\left([\nabla_l, \nabla_p]E^i, \tilde \nabla_i \tilde
E^l \right).
\end{aligned}
\end{eqnarray*}
Let us denote the sum of the last two terms on the right-hand side
by $\varpi_p$. Then
\begin{eqnarray*}
\varpi_p&=& g^{i k}\ast R^r_{k p l}\ast \Gamma^l_{r
i}-\tilde\Gamma^i_{l r}\ast g^{r k}\ast R^l_{k p i}.
\end{eqnarray*}
In the commutative case, $\varpi_p$ vanishes identically for all
$p$. However in the noncommutative setting, there is no reason to
expect this to happen. Let us now define
\begin{eqnarray}
\begin{aligned}
R^i_{p;i} &=& \partial_i R^i_p - \tilde\Gamma^i_{p r}\ast R^r_i
+ \tilde\Gamma^i_{i r}\ast R^r_p, \\
\Theta^l_{p;l} &=& \partial_l \Theta^l_p -
\Theta^r_l\ast\Gamma^l_{r p}+ \Theta^r_p\ast\Gamma^l_{r
l}&-\varpi_p.
\end{aligned}
\end{eqnarray}
Then the second Bianchi identity implies
\begin{eqnarray}\label{divergence}
R^i_{p;i} + \Theta^i_{p;i} - \partial_p R =0.
\end{eqnarray}

The above discussions suggest that Einstein's equation no longer
takes its usual form in the noncommutative setting. Instead,
formulae \eqref{divergence} and \eqref{Scalar2} suggest that
the following is a reasonable proposal for a noncommutative Einstein
equation in the vacuum:
\begin{eqnarray} \label{vacuum}
R^i_j + \Theta^i_j- \delta^i_j R &=&0.
\end{eqnarray}

We have not been able to formulate a basic principle which enables us to {\em
derive} \eqref{Einstein}. However, in the next section,
we shall solve this equation to obtain a class of exact
solutions. The existence of such solutions is evidence that it
is a meaningful candidate for a noncommutative Einstein field equation.

We may extend \eqref{vacuum} to include
matter and dark energy. We propose the following equation,
\begin{eqnarray} \label{Einstein}
R^i_j + \Theta^i_j- \delta^i_j R +2\delta ^i _j \Lambda =2T^i _j,
\end{eqnarray}
where $T^i _j$ is some generalized ``energy-momentum tensor",  and
$\Lambda$ is the cosmological constant. This reduces to the vacuum
equation \eqref{vacuum} when $T^i_j=0$ and the
cosmological constant vanishes. We hope to provide a mathematical
justification for this proposal in future work, where the defining
properties of $T^i _j$ will also be specified.

\subsection{Exact solutions in the vacuum}

We now construct a class of exact solutions of the
noncommutative vacuum Einstein field equations. The solutions are
quantum deformed analogues of plane-fronted gravitational waves
\cite{Br, ER, Ro, EK, Pe, Col}.

Let $(\theta _{ij})$ be an arbitrary constant skew symmetric $4
\times 4$ matrix, and endow the space of functions of the variables
$(x, y, u , v)$ with the Moyal product defined with respect to
$(\theta _{ij})$. We denote the resulting noncommutative algebra by
$\cA$.

Now we consider a noncommutative space $X$ embedded in $\cA^6$ by a
map of the form
\begin{eqnarray}\label{embedding}
X =\left(x, y, \frac{Hu+u+v}{\sqrt{2}},
\frac{H-\frac{u^2}{2}}{\sqrt{2}}, \frac{H u-u+v}{\sqrt{2}},
\frac{H-\frac{u^2}{2}}{\sqrt{2}}\right),
\end{eqnarray}
where, needless to say, the component functions are elements of
$\cA$. Here $H$ is an unknown function, which we shall determine by requiring
noncommutative space to be Einstein.

Let us take $\eta=diag(1, 1, 1, 1, -1, -1)$, and construct the
noncommutative metric $\g$ by using the formula \eqref{metric} for
this embedded noncommutative space. Denote
\[
\begin{aligned}
A=\theta _{yu} H_{xy}+\theta _{xu} H_{xx},
\quad B=\theta _{yu} H_{yy}+\theta _{xu} H_{xy},
\end{aligned}
\]
where $H_{x y}$ etc. are second order partial derivatives of $H$,
and $\theta_{x y}$ etc. refer to components of the matrix $\theta$.
A very lengthy calculation yields the following result for the
noncommutative metric:
\begin{eqnarray*}
\begin{aligned}
\g=&
\begin{pmatrix}
 1 & 0 & -\barh A  & 0 \\
 0 & 1 & -\barh B & 0 \\
 \barh A &
 \barh B & 2 H & 1 \\
 0 & 0 & 1 & 0
\end{pmatrix}.
\end{aligned}
\end{eqnarray*}
It is useful to note that in the classical limit with all $\theta_{i
j}=0$,  the matrix diagonalises to $diag(1, 1, H+\sqrt{1+H^2},
H-\sqrt{1+H^2})$, thus has Minkowski signature. Further tedious
computations produce the following inverse metric:
\begin{eqnarray*}
\begin{aligned}
\g^{-1}=
\begin{pmatrix}
 1 & 0 & 0 & \barh A  \\
 0 & 1 & 0 & \barh B \\
 0 & 0 & 0 & 1\\
 -\barh A  &
 -\barh B & 1 & g^{44}
\end{pmatrix}
\end{aligned}
\quad\text{with}\quad g^{44}=-\barh^2 (B\ast B + A\ast A) -2 H.
\end{eqnarray*}
Using these formulae we can compute $\Gamma _{ijk}$ and $\Gamma^k
_{ij}$, the nonzero components of which are given below:
\begin{eqnarray*}
\begin{aligned}
\Gamma_{113}=&-\barh(\theta _{yu} H_{xxy}+\theta _{xu}H_{xxx}),\\
\Gamma_{123}=&\Gamma_{213}= -\barh (\theta _{yu} H_{xyy}+\theta _{xu} H_{xxy}),\\
\Gamma_{133}=&\Gamma_{313}= H_x-\barh (\theta _{yu} H_{xyu}+\theta _{xu} H_{xxu}),\\
\Gamma_{223}=&-\barh (\theta _{yu} H_{yyy}+\theta _{xu} H_{xyy}),\\
\Gamma_{233}=&\Gamma_{323}=H_y-\barh (\theta _{yu} H_{yyu}+\theta _{xu} H_{xyu})\\
\Gamma_{331}=&-H_x, \quad \Gamma_{332}=-H_y,\\
\Gamma_{333}=&H_u-\barh(\theta _{yu} H_{yuu}+\theta_{xu} H_{xuu});\\
\Gamma_{11}^4=&-\barh (\theta _{yu} H_{xxy}+\theta _{xu}H_{xxx}),\\
\Gamma_{12}^4=&\Gamma_{21}^4= -\barh (\theta _{yu} H_{xyy}+\theta _{xu} H_{xxy})\\
\Gamma_{13}^4=&\Gamma_{31}^4=H_x-\barh (\theta _{yu} H_{xyu}+\theta _{xu} H_{xxu})\\
\Gamma_{22}^4=&-\barh(\theta _{yu} H_{yyy}+\theta _{xu} H_{xyy}),\\
\Gamma_{23}^4=&\Gamma_{32}^4=H_y-\barh (\theta _{yu} H_{yyu}+\theta _{xu} H_{xyu}),\\
\Gamma_{33}^1=&-H_x,\quad
\Gamma_{33}^2=-H_y,\\
\Gamma_{33}^4=&-H_x \ast \barh (\theta _{yu} H_{xy}+\theta_{xu} H_{xx})
-H_y \ast\barh (\theta _{yu}H_{yy}+\theta _{xu} \ H_{xy})\\
&+H_u-\barh (\theta _{yu} H_{yuu}+\theta _{xu}H_{xuu}).
\end{aligned}
\end{eqnarray*}

Remarkably, explicit formulae for curvatures can also be obtained,
even though the noncommutativity of the $\ast$-product complicates
the computations enormously. We have
\begin{eqnarray*}
\begin{array}{lll}
&R_{1313}=-R_{1331}=-H_{xx}, & R_{1323}=-R_{1332}=-H_{xy},\\
&R_{2313}=-R_{2331}=-H_{xy}, & R_{2323}=-R_{2332}=-H_{yy},\\
&R_{3113}=-R_{3131}=H_{xx},  & R_{3123}=-R_{3132}=H_{xy},\\
&R_{3213}=-R_{3231}=H_{xy},  & R_{3223}=-R_{3232}=H_{yy},\\
&R_{3331}=-R_{3313}=0,       & R_{3332}=-R_{3323}=0.
\end{array}
\end{eqnarray*}

Thus the nonzero components of $R^l _{ijk}$ are
\begin{eqnarray*}
\begin{array}{lll}
&R_{113}^4=-R_{131}^4=H_{xx}, \;\;\;\;\;\;\;\;R_{123}^4=-R_{132}^4=H_{xy},\\
&R_{213}^4=-R_{231}^4=H_{xy}, \;\;\;\;\;\;\;\;R_{223}^4=-R_{232}^4=H_{yy},\\
&R_{313}^1=-R_{331}^1=-H_{xx},\;\;\;\;\;R_{313}^2=-R_{331}^2=-H_{xy},\\
&R_{323}^1=-R_{332}^1=-H_{xy},\;\;\;\;\;R_{323}^2=-R_{332}^2=-H_{yy},\\
&R_{313}^4=-R_{331}^4=-H_{xx}\ast\barh (\theta _{yu}H_{xy}+\theta _{xu} H_{xx})\\
&\;\;\;\;\;\;\;\;\;\;\;\;\;\;\;\;\;\;\;\;\;\;\;\;\;\;
-H_{xy}\ast\barh (\theta_{yu} H_{yy}+\theta _{xu} H_{xy}),\\
&R_{323}^4=-R_{332}^4=-H_{xy}\ast\barh (\theta _{yu}H_{xy}+\theta _{xu} H_{xx})\\
&\;\;\;\;\;\;\;\;\;\;\;\;\;\;\;\;\;\;\;\;\;\;\;\;\;\;
-H_{yy}\ast\barh (\theta _{yu} H_{yy}+\theta _{xu}H_{xy}).
\end{array}
\end{eqnarray*}
From these formulae, we obtain the nonzero components of the Ricci
curvature:
\begin{eqnarray}\label{Equation}
R_3^4=\Theta_3^4=-H_{xx}-H_{yy}.
\end{eqnarray}

Thus the noncommutative vacuum Einstein field equations
(\ref{vacuum}) are satisfied if and only if the following equation
holds:
\begin{eqnarray}
H _{xx} +H _{yy} =0. \label{vacuum-pp}
\end{eqnarray}

Solutions of this linear equation for $H$ exist in abundance. Each
solution leads to an exact solution of the noncommutative vacuum
Einstein field equations. If we set $\theta$ to zero, we recover
from such a solution the plane-fronted gravitational wave \cite{Br,
ER, Ro, EK} in classical general relativity. Thus we shall call such
a solution of \eqref{vacuum} a {\em plane-fronted noncommutative
gravitational wave}.

It is clear from \eqref{Equation} that plane-fronted noncommutative
gravitational waves satisfy the additivity property. Explicitly, if
the noncommutative metrics of
\begin{eqnarray*}
X^{(i)} =\left(x, y, \frac{H_iu+u+v}{\sqrt{2}},
\frac{H_i-\frac{u^2}{2}}{\sqrt{2}}, \frac{H_iu-u+v}{\sqrt{2}},
\frac{H_i-\frac{u^2}{2}}{\sqrt{2}}\right), \quad i=1, 2,
\end{eqnarray*}
are plane-fronted noncommutative gravitational waves, we let
$H=H_1+H_2$, and set
\begin{eqnarray*}
X=\left(x, y, \frac{Hu+u+v}{\sqrt{2}},
\frac{H-\frac{u^2}{2}}{\sqrt{2}}, \frac{H u-u+v}{\sqrt{2}},
\frac{H-\frac{u^2}{2}}{\sqrt{2}}\right).
\end{eqnarray*}
Then the noncommutative metric of $X$ is also a plane-fronted
noncommutative gravitational wave. This is a rather nontrivial fact
since the noncommutative Einstein field equations are highly
nonlinear in $\g$, and it is extremely rare to have this additivity
property.

At this point, it is appropriate to point out that the embedding
\eqref{embedding} is only used as a device for constructing the
metric and the connection, from which the curvatures are derived.
However, we should observe the power of embeddings in solving the
noncommutative Einstein field equations. Without using the embedding
\eqref{embedding}, it would be very difficult to come up with
elegant solutions like what we have obtained here.

\section{Quantum spacetimes}\label{sect:spacetime}

In this section, we consider quantisations of several well known spacetimes.
We first find a global embedding of a spacetime into some pseudo-Euclidean
space, whose existence is guaranteed by theorems of Nash, Clarke and
Greene \cite{N, C, Gr}. Then we quantise the spacetime following the
strategy of deformation quantisation \cite{BFFLS, Ko} by deforming
\cite{Ge} the algebra of functions in the pseudo-Euclidean space to
the Moyal algebra. Through this mechanism, classical spacetime
metrics will deform to ``quantum'' noncommutative metrics which
acquire quantum fluctuations. In particular, certain anti-symmetric
components arise in the deformed metrics, which involve the Planck
constant and vanish in the classical limit.

\subsection{Quantum deformation of the Schwarzschild spacetime}\label{schwarzschild}

In this section, we investigate noncommutative analogues of the
Schwarzschild spacetime using the general theory discussed in
previous sections.  Recall that the Schwarzschild spacetime has the
following metric
\begin{eqnarray}\label{Sch}
d s^2=-\left(1-\frac{2m}{r}\right)dt^2
+\left(1-\frac{2m}{r}\right)^{-1}dr^2 +r^2\left(d\theta ^2 +\sin^2
\theta d\phi ^2 \right)
\end{eqnarray}
where $m=\frac{2 G M}{c^2}$ is constant, with $M$ interpreted as the
total mass of the spacetime. In the formula for $m$, $G$ is the
Newton constant, and $c$ is the speed of light. The Schwarzschild
spacetime can be embedded into a flat space of 6-dimensions in the
following two ways \cite{K1, K2, F}:

\medskip
\noindent(i). {Kasner's embedding:}
\begin{eqnarray*}
\begin{aligned}
& X^1=\left(1-\frac{2m}{r}\right)^\frac{1}{2} \sin t,\quad
X^2=\left(1-\frac{2m}{r}\right)^\frac{1}{2} \cos t,\\
&X^3=f(r),\quad (f')^2 +1=\left(1-\frac{2m}{r}\right)^{-1}
\left(1+\frac{m^2}{r^4}\right),\\
&X^4=r \sin\theta \cos\phi,\quad X^5=r \sin\theta \sin\phi,\quad
X^6=r \cos\theta,
\end{aligned}
\end{eqnarray*}
with the Schwarzschild metric given by
\begin{eqnarray*}
ds^2=-\big(dX^1\big)^2-\big(dX^2\big)^2+\big(dX^3\big)^2+\big(dX^4\big)^2
+\big(dX^5\big)^2+\big(dX^6\big)^2
\end{eqnarray*}

\medskip
\noindent(ii). {Fronsdal's embedding:}
\begin{eqnarray*}
\begin{aligned}
&Y^1=\left(1-\frac{2m}{r}\right)^\frac{1}{2} \sinh t,\quad
Y^2=\left(1-\frac{2m}{r}\right)^\frac{1}{2} \cosh t,\\
&Y^3=f(r),\quad (f')^2 +1=\left(1-\frac{2m}{r}\right)^{-1}
\left(1-\frac{m^2}{r^4}\right),\\
&Y^4=r \sin\theta \cos\phi,\quad Y^5=r \sin\theta \sin\phi,\quad
Y^6=r \cos\theta,
\end{aligned}
\end{eqnarray*}
with the Schwarzschild metric given by
\begin{eqnarray*}
ds^2=-\big(dY^1\big)^2+\big(dY^2\big)^2+\big(dY^3\big)^2+\big(dY^4\big)^2
+\big(dY^5\big)^2+\big(dY^6\big)^2.
\end{eqnarray*}

Let us now construct a noncommutative analogue of the Schwarzschild
spacetime. Denote $x^0=t$, $x^1=r$, $x^2=\theta$ and $x^3=\phi$. We
deform the algebra of functions in these variables by imposing on it
the Moyal product defined by \eqref{multiplication} with the
following anti-symmetric matrix
\begin{eqnarray}\label{asy-matrix}
\left(\theta _{\mu \nu}\right)_{\mu,
\nu=0}^3=\left(\begin{array}{cccc}
   0&  0&  0&  0\\
   0&  0&  0&  0\\
   0&  0&  0&  1\\
   0&  0&  -1&  0
\end{array}\right).
\end{eqnarray}
Denote the resultant noncommutative algebra by $\cA$. Note that in
the present case that the nonzero components of the matrix
$\left(\theta _{\mu \nu}\right)$ are dimensionless.

Now we regard the functions $X^i$, $Y^i$  $(1\le i\le 6)$ appearing
in both Kasner's and Fronsdal's embeddings as elements of $\cA$. For
$\mu=0, 1, 2, 3$, and $i=1, 2, \dots, 6$, let
\begin{eqnarray}\label{veibein}
\begin{aligned}
E^i_\mu&=& \frac{\partial X^i}{\partial x^\mu},
&\quad \text{for Kasner's embedding}, \\
E^i_\mu&=& \frac{\partial Y^i}{\partial x^\mu},
&\quad \text{for
Fronsdal's embedding}.
\end{aligned}
\end{eqnarray}
Following the general theory of the last section, we define the
metric and noncommutative torsion for the noncommutative
Schwarzschild spacetime by,

\medskip
\noindent (1). in the case of Kasner's embedding
\begin{eqnarray}\label{Kasner-case}
\begin{aligned}
g_{\mu \nu}=& -E^1_\mu\ast E^1_\nu - E^2_\mu\ast E^2_\nu +
\sum_{j=3}^6 E^j_\mu\ast E^j_\nu,\\
\Upsilon_{\mu\nu\rho} =& \frac{1}{2}\left(-\partial_\mu E^1_\nu\ast
E^1_\rho -  \partial_\mu E^2_\nu\ast E^2_\rho +\sum_{j=3}^6
\partial_\mu E^j_\nu\ast
E^j_\rho \right)\\
&+\frac{1}{2}\left(-E^1_\rho\ast  \partial_\mu E^1_\nu -E^2_\rho\ast
\partial_\mu E^2_\nu+\sum_{j=3}^6
E^j_\rho  \ast \partial_\mu E^j_\nu\right);
\end{aligned}
\end{eqnarray}
\noindent (2). in the case of Fronsdal's embedding
\begin{eqnarray}\label{Fronsdal-case}
\begin{aligned}
g_{\mu \nu}=& -E^1_\mu\ast E^1_\nu + \sum_{j=2}^6 E^j_\mu\ast
E^j_\nu, \\
\Upsilon_{\mu\nu\rho} =& \frac{1}{2}\left(-\partial_\mu E^1_\nu\ast
E^1_\rho +\sum_{j=2}^6
\partial_\mu E^j_\nu\ast
E^j_\rho \right)\\
&+\frac{1}{2}\left(-E^1_\rho\ast  \partial_\mu E^1_\nu +\sum_{j=2}^6
E^j_\rho  \ast \partial_\mu E^j_\nu\right).
\end{aligned}
\end{eqnarray}

Some lengthy but straightforward calculations show that the metrics
and the noncommutative torsions are respectively equal in the two
cases. Since the noncommutative torsion will not be used in later
discussions, we shall not spell it out explicitly. However, we
record the metric $\g=(g_{\mu \nu})$ of the quantum deformation of
the Schwarzschild spacetime below:
\begin{eqnarray}\label{deformed-Sch}
\begin{aligned}
g_{0 0} =&-\left(1-\frac{2m}{r}\right),\\
g_{0 1} =&g_{1 0}=g_{0 2} =g_{2 0}=g_{0 3} =g_{3 0}=0, \\
g_{1 1}=&\left(1-\frac{2m}{r}\right)^{-1}
\left[1+\left(1-\frac{2m}{r}\right)
          \left(\sin^2\theta -\cos^2 \theta\right)\sinh^2\barh\right],\\
g_{1 2} =    &g_{2 1} = 2r\sin\theta\cos\theta\sinh^2\barh,\\
g_{1 3} =    &-g_{3 1} = -2r\sin\theta\cos\theta\sinh\barh\cosh\barh,\\
g_{2 2} =    &r^2
    \left[1-\left(\sin^2\theta-\cos^2\theta\right)\sinh^2\barh\right],\\
g_{2 3} = &-g_{3 2}
=r^2\left(\sin^2\theta-\cos^2\theta\right)\sinh\barh\cosh\barh,\\
g_{3 3} =&r^2
    \left[\sin^2\theta+\left(\sin^2\theta-\cos^2\theta\right)\sinh^2\barh\right].
\end{aligned}
\end{eqnarray}

It is interesting to observe that the quantum deformation of the
Schwarzschild metric (\ref{deformed-Sch}) still has a black hole
with the event horizon at $r=2m$. The Hawking temperature and
entropy of the black hole are respectively given by
\begin{eqnarray*}
\begin{aligned}
T=\frac{1}{2}\frac{dg_{00}}{dr}\Big|_{r=2m}=\frac{1}{4m},\qquad
S_{bh}=4\pi m^2.
\end{aligned}
\end{eqnarray*}
They coincide with the temperature and entropy of the classical
Schwarzschild black hole of mass $M$. However, the area of the event
horizon of the noncommutative black hole receives corrections from
the quantum deformation of the spacetime. Let
$\bar{g}=\begin{pmatrix}g_{2 2} & g_{2 3}\\ g_{3 2} & g_{3 3}
\end{pmatrix}.$ We have
\begin{eqnarray*}
\begin{aligned}
A=&\iint _{\{r=2m\}} \sqrt{\det\bar{g}}\, d\theta d\phi\\
=&\iint _{\{r=2m\}} r^2 \sin\theta \sqrt{1+\left(\sin^2\theta
-\cos^2\theta\right)\sinh^2 \barh}d\theta d\phi\\
=&16\pi m^2 \left(1-\frac{\barh^2}{6} +O(\barh^4)\right).
\end{aligned}
\end{eqnarray*}
This leads to the following relationship between the horizon area
and entropy of the noncommutative black hole:
\begin{eqnarray}\label{entropy}
S_{bh} =\frac{A}{4}\left(1+\frac{\barh^2}{6} +O(\barh^4)\right).
\end{eqnarray}

Let us now consider the Ricci and $\Theta$-curvature of the deformed
Schwarzschild metric. We have
\begin{eqnarray*}\label{RTheta}
\begin{aligned}
R_0^1=&R_0^2=R_0^3=R_1^0=R_2^0=R_3^0=0,\\
\Theta_0^1=&\Theta_0^2=\Theta_0^3=\Theta_1^0=\Theta_2^0=\Theta_3^0=0,\\
R_0 ^0 =&\Theta _0 ^0=-\frac{m\left[ 2m + 3r + 3\left( m + r
\right)\cos2\theta\right]}{r^4}\barh^2 +O(\barh^4),\\
R_1 ^1 =&\Theta _1 ^1=\frac{m\left[-14m + 3r + \left( -11m + r
\right)\cos2\theta\right]}{r^4} \barh^2+ O(\barh^4),\\
R_1 ^2 =&\Theta _1 ^2=\frac{2m\cos^2\theta \cot\theta}{r^4} \barh^2+
O(\barh^4),\\
R_1 ^3 =&-\Theta _1 ^3=\frac{2m\cot\theta}{r^4}\barh +O(\barh^3),\\
R_2 ^1 =&\Theta _2 ^1=\frac{5m\left( -2m + r \right)\sin
2\theta}{r^3} \barh^2+ O(\barh^4),\\
R_2 ^2 =&\Theta _2 ^2=\frac{m\left[4\left(m + r \right)  + \left( 6m
+ 5r \right)\cos2\theta \right]}{r^4}\barh^2 + O(\barh^4),\\
R_2 ^3 =&-\Theta _2 ^3=\frac{4m}{r^3}\barh +O(\barh^3),\\
R_3 ^1 =&-\Theta _3 ^1=\frac{m\left( 2m - r
\right)\sin2\theta}{r^3}\barh+O(\barh^3).\\
R_3 ^2 =&-\Theta _3 ^2=\frac{4m \cos^2\theta}{r^3}\barh+O(\barh^3).\\
R_3 ^3 =&\Theta _3 ^3=\frac{m\left[-8m +8r+\left( -6m + 9r \right)
\cos2\theta \right]}{r^4}\barh^2 + O(\barh^4).
\end{aligned}
\end{eqnarray*}

Note that $R_i^i=\Theta_i^i$ for all $i$, and $R_i^j=-\Theta_i^j$ if
$i\ne j$. Let us write
\begin{eqnarray}\label{expansion}
\begin{aligned}
R_j^i= {R_j^i}_{(0)} + \barh {R_j^i}_{(1)} + \barh^2
{R_j^i}_{(2)}+\dots, \\
\Theta_j^i= {\Theta_j^i}_{(0)} + \barh {\Theta_j^i}_{(1)} + \barh^2
{\Theta_j^i}_{(2)}+\dots. \end{aligned}
\end{eqnarray}
Then the formulae for $R_j^i$ and $\Theta_j^i$ show that
\[
{R_j^i}_{(0)}={\Theta_j^i}_{(0)}, \quad
{R_j^i}_{(1)}=-{\Theta_j^i}_{(1)}, \quad {R_j^i}_{(2)}
={\Theta_j^i}_{(2)}.
\]

Naively generalizing the Einstein tensor $R^i _j-
\frac{1}{2}\delta^i_j R$ to the noncommutative setting, one ends up
with a quantity that does not vanish at order $\barh$, as can be
easily shown using the above results. However,
\[ R^i _j +\Theta ^i _j
-\delta ^i _j R =0 + O(\barh^2).
\]
This indicates that the proposed noncommutative Einstein equation
\eqref{Einstein} captures some essence of the underlying symmetries
in the noncommutative world.

Now the deformed Schwarzschild metric (\ref{deformed-Sch}) satisfies
the vacuum noncommutative Einstein equation (\ref{Einstein}) with
$T^i _j=0$ and $\Lambda=0$ to first order in the deformation
parameter. However, if we take into account higher order corrections
in $\barh$, the deformed Schwarzschild metric no longer satisfies
the noncommutative Einstein equation in the vacuum. Instead, $R^i _j
+\Theta ^i _j -\delta ^i _j R = T^i_j$ with $T^i_j$ being of order
$O(\barh^2)$ and given by
\begin{eqnarray}\label{T}
\begin{aligned}
T_0^1=&T_0^2=T_0^3=T_1^0=T_2^0=T_3^0=T_3^1=T_3^2=T_1^3=T_2^3=0,\\
T_0^0=&\frac{m\left[ 8m - 9r + \left( 4m - 9r
\right)\cos2\theta\right]}{r^4}\barh^2 + O(\barh^4),\\
T_1^1=&\frac{-m\left[ 4m + 3r + \left( 4m + 5r
\right)\cos2\theta\right]}{r^4}\barh^2 + O(\barh^4),\\
T_1^2=&\frac{2m \cos^2 \theta \cot\theta }{r^4}\barh^2 +O(\barh^4),\\
T_2^1=&\frac{5m\left( -2m + r \right) \sin2\theta}{r^3}\barh^2 + O(\barh^4),\\
T_2^2=&\frac{m\left[ 14m - 2r + \left( 13m - r
\right)\cos2\theta\right]}{r^4}\barh^2 + O(\barh^4),\\
T_3^3=&\frac{m\left[2\left(m + r \right)  + \left( m + 3r
\right)\cos2\theta\right]}{r^4}\barh^2 +O(\barh^4).
\end{aligned}
\end{eqnarray}

A possible physical interpretation of the results is the following.
We regard the $\barh$ and higher order terms in the metric $g_{i j}$
and associated curvature $R_{i j k l}$ as arising from quantum
effects of gravity. Then the $T_i^j$ obtained in \eqref{T} should be
interpreted as quantum corrections to the classical energy-momentum tensor.

\subsection{Quantum deformation of the Schwarzschild-de
Sitter spacetime}\label{schwarzschild-dS}

In this section, we investigate a noncommutative analogue of the
Schwarzschild-de Sitter spacetime. Since the analysis is parallel to
that on the quantum Schwarzschild spacetime, we shall only present
the pertinent results.

Recall that the Schwarzschild-de Sitter spacetime has the following
metric
\begin{eqnarray}\label{Sch-dS}
\begin{aligned}
ds^2=&-\left(1-\frac{r^2}{l^2}-\frac{2m}{r}\right)dt^2
+\left(1-\frac{r^2}{l^2}-\frac{2m}{r}\right)^{-1}dr^2\\
&+r^2\left(d\theta ^2 +\sin^2 \theta d\phi ^2 \right)
\end{aligned}
\end{eqnarray}
where $\frac{3}{l^2}=\Lambda>0$ is the cosmological constant, and
$m$ is related to the total mass of the spacetime through the same
formula as in the Schwarzschild case. This spacetime can be embedded
into a flat space of 6-dimensions in two different ways.

\medskip
\noindent(i). {Generalized Kasner embedding:}
\begin{eqnarray*}
\begin{aligned}
&X^1=\left(1-\frac{r^2}{l^2}-\frac{2m}{r}\right)^\frac{1}{2} \sin t, \quad
X^2=\left(1-\frac{r^2}{l^2}-\frac{2m}{r}\right)^\frac{1}{2} \cos t,\\
&X^3=f(r),\quad
(f')^2+1=\left(1-\frac{r^2}{l^2}-\frac{2m}{r}\right)^{-1}
\left[1+\left(\frac{m}{r^2}-\frac{r}{l^2}\right)^2\right],\\
&X^4=r \sin\theta \cos\phi,\quad X^5=r \sin\theta \sin\phi,\quad
X^6=r \cos\theta,
\end{aligned}
\end{eqnarray*}
with the Schwarzschild-de Sitter metric given by
\begin{eqnarray*}
ds^2=-\big(dX^1\big)^2-\big(dX^2\big)^2+\big(dX^3\big)^2+\big(dX^4\big)^2
+\big(dX^5\big)^2+\big(dX^6\big)^2
\end{eqnarray*}

\medskip
\noindent(ii). {Generalized Fronsdal embedding:}
\begin{eqnarray*}
\begin{aligned}
&Y^1=\left(1-\frac{r^2}{l^2}-\frac{2m}{r}\right)^\frac{1}{2} \sinh
t, \quad
Y^2=\left(1-\frac{r^2}{l^2}-\frac{2m}{r}\right)^\frac{1}{2} \cosh t,\\
&Y^3=f(r),\quad (f')^2
+1=\left(1-\frac{r^2}{l^2}-\frac{2m}{r}\right)^{-1}
\left[1-\left(\frac{m}{r^2}-\frac{r}{l^2}\right)^2\right],\\
&Y^4=r \sin\theta \cos\phi, \quad Y^5=r \sin\theta \sin\phi, \quad
Y^6=r \cos\theta,
\end{aligned}
\end{eqnarray*}
with the Schwarzschild-de Sitter metric given by
\begin{eqnarray*}
ds^2=-\big(dY^1\big)^2+\big(dY^2\big)^2+\big(dY^3\big)^2+\big(dY^4\big)^2
+\big(dY^5\big)^2+\big(dY^6\big)^2.
\end{eqnarray*}

Let us now construct a noncommutative analogue of the
Schwarzschild-de Sitter spacetime. Denote $x^0=t$, $x^1=r$,
$x^2=\theta$ and $x^3=\phi$. We deform the algebra of functions in
these variables by imposing on it the Moyal product defined by
\eqref{multiplication} with the anti-symmetric matrix
(\ref{asy-matrix}).  Denote the resultant noncommutative algebra by
$\cA$.

Now we regard the functions $X^i$, $Y^i$  $(1\le i\le 6)$ appearing
in both the generalized Kasner embedding  and the generalized
Fronsdal embedding as elements of $\cA$. Let $E_\mu^i$ ($\mu=0, 1,
2, 3$, and $i=1, 2, \dots, 6$) be defined by \eqref{veibein} but for
the generalized Kasner and  Fronsdal embeddings respectively. We
also define the metric and noncommutative torsion for the
noncommutative Schwarzschild-de Sitter spacetime by equations
\eqref{Kasner-case} and \eqref{Fronsdal-case} for the generalized
Kasner and Fronsdal embeddings respectively. As in the case of the
noncommutative Schwarzschild spacetime, we can show that the metrics
and the noncommutative torsions are respectively equal for the two
embeddings. We record the metric $\g=(g_{\mu \nu})$ of the quantum
deformation of the Schwarzschild-de Sitter spacetime below:
\begin{eqnarray}\label{deformed-Sch-dS}
\begin{aligned}
g_{0 0} =&-\left(1-\frac{r^2}{l^2}-\frac{2m}{r}\right),\\
g_{0 1} =&g_{1 0}=g_{0 2} =g_{2 0}=g_{0 3} =g_{3 0}=0, \\
g_{1 1}=&\left(1-\frac{r^2}{l^2}-\frac{2m}{r}\right)^{-1}
+\left(\sin^2\theta -\cos^2 \theta\right)\sinh^2\barh,\\
g_{1 2} =    &g_{2 1} = 2r\sin\theta\cos\theta\sinh^2\barh,\\
g_{1 3} =    &-g_{3 1} = -2r\sin\theta\cos\theta\sinh\barh\cosh\barh,\\
g_{2 2} =    &r^2
    \left[1-\left(\sin^2\theta-\cos^2\theta\right)\sinh^2\barh\right],\\
g_{2 3} = &-g_{3 2}
=r^2\left(\sin^2\theta-\cos^2\theta\right)\sinh\barh\cosh\barh,\\
{g}_{3 3} =&r^2
    \left[\sin^2\theta+\left(\sin^2\theta-\cos^2\theta\right)\sinh^2\barh\right].
\end{aligned}
\end{eqnarray}

Let us now consider the Ricci and $\Theta$ curvatures
of the deformed Schwarzschild metric. We have
\begin{eqnarray*}\label{RTheta-dS}
\begin{aligned}
R_0^1=&R_0^2=R_0^3=R_1^0=R_2^0=R_3^0=0,\\
\Theta_0^1=&\Theta_0^2=\Theta_0^3=\Theta_1^0=\Theta_2^0=\Theta_3^0=0,\\
R_0^0=&\Theta_0^0=\frac{3}{l^2} + \Big[ l^2 \left( 10m - 3r
\right)r^3 + 10r^6
- l^4 m \left( 2m + 3r \right)\\
&-3\left\{ -2 l^2 m r^3 - 2 r^6 + l^4 m
\left( m + r \right)  \right\} \cos 2\theta \Big] \frac{\barh^2}{l^4 r^4}
+O(\barh^4),\\
R_1^1=&\Theta_1^1=\frac{3}{l^2} + \Big[ l^2 \left( 16 m - 9r \right)
r^3 + 16
r^6 + l^4 m \left( -14 m + 3 r \right) \\
& + \left\{ 2 l^2 \left( 5 m - 2 r \right) r^3 + 10 r^6 + l^4 m
\left( -11 m + r \right) \right\} \cos 2\theta \Big]
\frac{\barh^2}{l^4 r^4} +O(\barh^4),\\
R_1^2=&\Theta_1^2=\frac{2 \left( l^2 m - 4 r^3 \right)
\cos^2 \theta \cot \theta}{l^2 r^4} \barh^2 +O(\barh^4),\\
R_1^3=&-\Theta_1^3=\frac{2\, \left(l^2 m - 4r^3 \right)
\cot\theta}{l^2 r^4}\barh +O(\barh^3),\\
R_2^1=&\Theta_2^1=-\frac{\left[ l^2 \left( 2m - r \right)  + r^3
\right]\left(5 l^2 m + 4 r^3 \right) \sin 2\theta}{l^4 r^3}\barh^2 +O(\barh^4),\\
\end{aligned}
\end{eqnarray*}
\begin{eqnarray*}
\begin{aligned}
R_2^2=&\Theta_2^2=\frac{3}{l^2} + \Big[ l^2 \left( 22 m - r \right)
r^3 + 10
r^6 + 4 l^4 m \left(m + r \right) \\
& + \left\{ 6r^6 + l^2 r^3 \left(15 m + 4 r \right)  + l^4 m \left(
6m + 5r \right)  \right\}\cos 2\theta \Big] \frac{\barh^2}{l^4 r^4} +O(\barh^4),\\
R_2^3=&-\Theta_2^3=\left( \frac{8}{l^2} + \frac{4m}{r^3} \right) \barh +O(\barh^3),\\
R_3^1=&-\Theta_3^1=\frac{\left( l^2 m - 4r^3 \right) \left[l^2 \left( 2m - r \right)
                   + r^3 \right] \sin 2\theta}{l^4 r^3} \barh+O(\barh^3),\\
R_3^2=&-\Theta_3^2=4 \left( \frac{2}{l^2} + \frac{m}{r^3} \right) \,
\cos^2\theta \barh +O(\barh^3),\\
R_3^3=&\Theta_3^3=\frac{3}{l^2} + \Big[ -8 l^4 m \left( m - r
\right) +l^2 \left( 28 m - 5r \right)r^3 + 16 r^6 \\
              &+ 3\left\{ 7 l^2 m r^3 + 4 r^6 + l^4 m
              \left( -2 m + 3r \right)  \right\} \cos 2\theta \Big]
              \frac{\barh^2}{l^4 r^4}
              +O(\barh^4).
\end{aligned}
\end{eqnarray*}
Note that if we expand $R^j_i$ and $\Theta^j_i$ into power series in
$\barh$ in the form \eqref{expansion}, we again have
\[
{R_j^i}_{(0)}={\Theta_j^i}_{(0)}, \quad
{R_j^i}_{(1)}=-{\Theta_j^i}_{(1)}, \quad {R_j^i}_{(2)}
={\Theta_j^i}_{(2)}.
\]

By using the above results one can easily show that the deformed
Schwarzschild-de Sitter metric (\ref{deformed-Sch-dS}) satisfies the
vacuum noncommutative Einstein equation (\ref{Einstein}) (with $T^i
_j=0$) to first order in the deformation parameter:

\[ R^i _j +\Theta ^i _j
-\delta ^i _j R +\delta^i _j \frac{6}{l^2}=0 + O(\barh^2).
\]

Further analysing the deformed Schwarzschild-de Sitter metric, we
note that $R^i _j +\Theta ^i _j -\delta ^i _j R +\delta^i _j
\frac{6}{l^2}=T^i_j$ with $T^i_j$ being of order $O(\barh^2)$ and
given by
\begin{eqnarray*}
\begin{aligned}
T_0^1=&T_0^2=T_0^3=T_1^0=T_2^0=T_3^0=T_3^1=T_3^2=T_1^3=T_2^3=0,\\
T_0^0=&- \Big[ 2 l^2 \left( 14m-3r \right)r^3 + 16 r^6 + l^4 m
\left( -8m + 9r \right) \\
& + \left\{ 20 l^2 m r^3 + 11 r^6 + l^4 m \left( -4m + 9r \right)
\right\}\cos2\theta \Big]\frac{\barh^2}{l^4 r^4} +O(\barh^4),\\
T_1^1=&- \Big[ 22 l^2 m r^3 + 10 r^6 + l^4 m \left(4 m + 3r
\right)\\&+ \left\{7r^6 + 4 l^2 r^3 \left( 4 m + r \right)  + l^4 m
\left( 4 m+ 5 r \right)  \right\} \cos 2\theta \Big]
\frac{\barh^2}{l^4 r^4}+O(\barh^4),\\
T_1^2=&\frac{2 \left( l^2 m - 4 r^3 \right) \cos^2\theta \cot\theta}
{l^2 r^4}\barh^2 +O(\barh^4),\\
\end{aligned}
\end{eqnarray*}
\begin{eqnarray*}
\begin{aligned}
T_2^1=&-\frac{\left[ l^2 \left( 2m - r \right)  + r^3 \right]\left(
5 l^2m + 4 r^3 \right) \sin2\theta}{l^4 r^3} \barh^2+O(\barh^4),\\
T_2^2=&-\Big[2 \left\{ 4 l^2 \left( 2m - r \right) r^3 + 8 r^6 + l^4
m \left( -7m + r \right)  \right\} \\
& + \left\{l^2 \left(11m - 4r \right)r^3 + 11 r^6 + l^4 m \left( -13
m + r \right)\right\}\cos 2\theta \Big] \frac{\barh^2}{l^4 r^4} +O(\barh^4),\\
T_3^3=& \Big[2 \left\{ -5 r^6 + l^4 m \left( m + r \right) + l^2 r^3
\left( -5m + 2r \right)  \right\} \\
& + \left\{ -5 l^2 m r^3 - 5 r^6 + l^4 m \left( m + 3r \right)
\right\} \cos 2\theta \Big] \frac{\barh^2}{l^4 r^4}+O(\barh^4).
\end{aligned}
\end{eqnarray*}
Similar to the case of the quantum Schwarzschild spacetime, one may
regard this as quantum corrections to the energy-momentum tensor.

\subsection{Noncommutative gravitational collapse}\label{Collapse}

Gravitational collapse is one of the most dramatic phenomena in the
universe. When the pressure is not sufficient to balance the
gravitational attraction inside a star, the star undergoes sudden
gravitational collapse possibly accompanied by a supernova
explosion, reducing to a super dense object such as a neutron star
or black hole.

In 1939, Oppenheimer and Snyder \cite{OS} investigated the collapse
process of ideal spherically symmetric stars equipped with the
Tolman metric \cite{T}. When the energy-momentum of an ideal star is
assumed to be given by perfect fluids, Tolman's metric allows the
case of dust which has zero pressure. In the dust case, Oppenheimer
and Snyder solved the Einstein field equations by further assuming
that the energy density is constant. They showed that stars above
the Tolman-Oppenheimer-Volkoff mass limit \cite{OS} (approximately
three solar masses) would collapse into black holes for reasons
given by Chandrasekhar. The work of Oppenheimer and Snyder also
marked the beginning of the modern theory of black holes.

The Tolman metric studied in \cite{OS} can be written as
\begin{eqnarray}
ds^2=-dt^2+(1-c t)^{4/3} \big[dr^2+r^2 (d\theta ^2+ \sin ^2\theta
d\phi ^2)\big] \label{Oppenhm-S}
\end{eqnarray}
with $c=3 r_0 ^\frac{1}{2} R _b ^{-\frac{3}{2}}$, where $r_0$ is the
gravitational radius and $R_b$ is the radius of the star. One may
examine the behaviour of the scalar curvature as time increases.
When time approaches the value $1/c$, the scalar curvature goes to
$\infty$, thus the radius of the stellar object reduces to zero. By
the reasoning of \cite{W}, this indicates gravitational collapse.
Obviously this only provides a snapshot, nevertheless, it enables
one to gain some understanding of gravitational collapse.

In this section, we quantise the dust solutions \cite{OS} and study
noncommutative gravitational collapse. Our method for quantisation
is much the same as in previous sections.

The Tolman spacetime can be embedded into a 5-dimensional flat
Minkowski spacetime via
\begin{eqnarray}\label{five-X}
\begin{aligned}
X^1=&\frac{9 (1-c t)^{4/3}}{32 c^2}+\Big(\frac{r^2}{4}+1\Big) (1-c t)^{2/3},\\
X^2=&\frac{9 (1-c t)^{4/3}}{32 c^2}+\Big(\frac{r^2}{4}-1\Big) (1-c t)^{2/3},\\
X^3=&(1-c t)^{2/3} r \cos \phi  \sin \theta ,\quad
X^4=(1-c t)^{2/3} r \sin \theta  \sin \phi ,\\
X^5=&(1-c t)^{2/3} r \cos \theta .
\end{aligned}
\end{eqnarray}

We deform the algebra of functions in the variables $r, t, \pi$ and
$\theta$ into a Moyal algebra ${\mathcal A}$ defined by the
anti-symmetric matrix
\begin{eqnarray}\label{ansatz}
\left(\theta^{\mu \nu}\right)_{\mu,
\nu=0}^3=\left(\begin{array}{cccc}
   0&  0&  0&  0\\
   0&  0&  0&  0\\
   0&  0&  0&  1\\
   0&  0&  -1&  0
\end{array}\right).
\end{eqnarray}
Now we consider the noncommutative geometry embedded in ${\mathcal
A}^5$ by \eqref{five-X}. The noncommutative metric of the embedded
noncommutative geometry (defined in the standard way \cite{CTZZ})
yields a quantum deformation of the metric (\ref{Oppenhm-S}):
\begin{eqnarray}
\begin{aligned}
g_{\mu \nu} =& -\partial_\mu X^1\ast \partial_\nu X^1 +\partial_\mu
X^2\ast \partial_\nu X^2
+\partial_\mu X^3\ast \partial_\nu X^3  \\
&+ \partial_\mu X^4\ast \partial_\nu X^4 +\partial_\mu X^5\ast
\partial_\nu X^5,
\end{aligned}
\end{eqnarray}
which can be computed explicitly. We have
\begin{eqnarray*}
\begin{aligned}
g_{11}=&-\frac{4 c^2 r^2 \cos  2 \theta  \sinh ^2\bar h}{9 (1-c t)^{2/3}}-1,\\
g_{12}=&g_{21}=\frac{2}{3} c r (1-c t)^{1/3}\cos  2 \theta  \sinh ^2\bar h,\\
g_{13}=&g_{31}=-\frac{4}{3} c r^2 (1-c t)^{1/3} \cos  \theta  \sin  \theta  \sinh ^2\bar h,\\
g_{14}=&-g_{41}=\frac{1}{3} c r^2 (1-c t)^{1/3} \sin  2 \theta
\sinh  2 \bar h,\\
g_{22}=&(1-c t)^{4/3} \left(1-\cos  2 \theta  \sinh ^2\bar h\right),\\
g_{23}=&g_{32}=r (1-c t)^{4/3} \sin  2 \theta  \sinh ^2\bar h,\\
g_{24}=&-g_{42}=-2 r (1-c t)^{4/3} \cos  \theta  \cosh \bar h \sin  \theta  \sinh \bar h,\\
g_{33}=&r^2 (1-c t)^{4/3} \left(\cos  2 \theta  \sinh ^2\bar h+1\right),\\
\end{aligned}
\end{eqnarray*}
\begin{eqnarray*}
\begin{aligned}
g_{34}=&-g_{43}=-\frac{1}{2} r^2 (1-c t)^{4/3} \cos  2 \theta  \sinh  2 \bar h,\\
g_{44}=&-\frac{1}{2} r^2 (1-c t)^{4/3} (\cos  2 \theta  \cosh  2
\bar h-1).
\end{aligned}
\end{eqnarray*}
The noncommutative scalar curvature is given by
\begin{eqnarray}\label{Cs}
R=\frac{4 c^2 \cosh ^2\bar h}{(1-c t)^{4/3}}\frac{C_1}{C^3}
\end{eqnarray}
where $C$ and $C_1$ are the following functions
\begin{eqnarray*}
\begin{aligned}
C=&9(1-c t)^{2/3} \cosh ^4\bar h-2 c^2 r^2 (2 \cos  2 \theta +\cosh
2 \bar h+3) \sinh ^2\bar h,
\end{aligned}
\end{eqnarray*}
\begin{eqnarray*}
\begin{aligned}
C_1 =&-243 (1-c t)^{4/3} \cosh ^8\bar h+486 (1-c t)^{4/3} \cosh ^6\bar h\\
   &-18 c^2 r^2 (1-c t)^{2/3} (2 \cos  2 \theta -3 \cosh  2 \bar h-1) \sinh ^2\bar h \cosh^4\bar h\\
   &-9 c^2 r^2 (1-c t)^{2/3} \Big(52 \cosh  2 \bar h+3 \cosh  4 \bar h\\
   &+\cos  2 \theta  (28 \cosh  2 \bar h+\cosh  4 \bar h-13)+9\Big) \sinh ^2\bar h \cosh ^2\bar h\\
   &+4 c^4 r^4  \sinh^4\bar h \Big(4 \cos  2 \theta  (\cosh  2 \bar h+15) \sinh ^2\bar h\\
   &+2 \cos  4 \theta  (\cosh  2 \bar h-3)+38 \cosh  2 \bar h+3 \cosh  4 \bar h-13\Big).
\end{aligned}
\end{eqnarray*}

Let us regard $\bar h$ as a real number and make the (physically
realistic) assumption that $\bar h$ is positive but close to zero.
Now if $t$ is significantly smaller than $\frac{1}{c}$ compared to
$\bar h$, that is, $\frac{1}{c}-t\gg\bar h$, both the noncommutative
metric and noncommutative scalar curvature $R$ are finite, and there
is non-singularity in the noncommutative spacetime. Thus the stellar
object described by the noncommutative geometry behaves much the
same as the corresponding classical object.

When $t=t_*:=\frac{1}{c}$, we have $R | _ {t_{*}
=\frac{1}{c}}=\infty$ and the radius of the stellar object reduces
to zero. This is the time when gravitational collapse happens in the
usual classical setting.

However, in the noncommutative case, singularities of the scalar
curvature already appear before $t_*$. Indeed, when time reaches
\[
\begin{aligned}
t(r, \theta)&=\frac{1}{c} - \frac{\sqrt{8}}{27} c^2 r^3
(2 \cos  2 \theta +\cosh  2 \bar h+3)^{3/2} \frac{\sinh^3\bar h}{\cosh^6\bar h}\\
&\cong \frac{1}{c} - \frac{8}{27} c^2 r^3 (\cos  2 \theta + 2)^{3/2}
\bar h^3,
\end{aligned}
\]
$C$ vanishes and $\frac{C_1}{(1-c t)^{4/3}}$ is finite of order $0$
in $\bar h$. Thus the scalar curvature tends to infinity for all
$t(r, \theta)$ and the noncommutative spacetime becomes singular.
Therefore, gravitational collapse happens within a certain range of
time because of the quantum effects captured by the noncommutativity
of spacetime. However, effect of noncommutativity only starts to
appear at third order of $\bar h$.


\bigskip


\begin{thebibliography}{9999}

\bibitem{AMV} L. \'Alvarez-Gaum\'e, F. Meyer,  M. A. Vazquez-Mozo,
Comments on noncommutative gravity. Nucl. Phys. {\bf B 75} (2006), 392.

\bibitem{ANSS} S. Ansoldi, P. Nicolini, A. Smailagic and E. Spallucci,
Noncommutative geometry inspired charged black holes. Phys. Lett. {\bf B 645} (2007), 261.

\bibitem{ADMW1}  P. Aschieri, C. Blohmann, M. Dimitrijevic, F. Meyer, P. Schupp, J. Wess,
A gravity theory on noncommutative spaces.  Class. Quant. Grav. {\bf 22} (2005), 3511.

\bibitem{ADMW2} P. Aschieri, M. Dimitrijevic, F. Meyer, J. Wess, Noncommutative
geometry and gravity. Class. Quant. Grav. {\bf 23} (2006), 1883.

\bibitem{BNK} R. Banerjee, B. R. Majhi and S. K. Modak, Area law in noncommutative
Schwarzschild black hole. Class. Quantum Grav. {\bf 26}, 085010 (2009), 11 pp.

\bibitem{BFFLS} F. Bayen, M. Flato, C. Fronsdal, A. Lichnerowicz and D. Sternheimer,
Quantum mechanics as a deformation of classical mechanics. Lett. Math. Phys. {\bf 1} (1977), 521-530.

\bibitem{Br} H. W. Brinkmann, Einstein spaces which are mapped conformally on each other.
Math. Ann. {\bf 94} (1925), 119-145.


\bibitem{CKNT} M. Chaichian, P. P. Kulish, K. Nishijima, A. Tureanu,  On a Lorentz-invariant
interpretation of noncommutative space-time and its implications on noncommutative QFT.
Phys. Lett. {\bf B 604} (2004), 98.

\bibitem{CKTZZ} M. Chaichian, P. P. Kulish, A. Tureanu, R. B. Zhang, X. Zhang,
Noncommutative fields and actions of twisted Poincar\'e algebra. J. Math. Phys. {\bf 49}, 042302 (2008), 16 pp.

\bibitem{COTZ} M. Chaichian, M. Oksanen, A. Tureanu, G. Zet,  Gauging the twisted
Poincare symmetry as noncommutative theory of gravitation,
Phys. Rev. {\bf D 79}, 044016 (2009),  8 pp.

\bibitem{CPT} M. Chaichian, P. Pre\v{s}najder, A. Tureanu, New concept of
relativistic invariance in noncommutative space-time: Twisted Poincare symmetry and its implications.
Phys. Rev. Lett. {\bf 94}
(2005), 151602.

\bibitem{CTZ} M. Chaichian, A. Tureanu, G. Zet, Twist as a symmetry principle and
the noncommutative gauge theory formulation. Phys. Lett. {\bf B 651} (2007), 319-323.

\bibitem{CTZZ} M. Chaichian, A. Tureanu, R. B. Zhang, X. Zhang,
Riemannian geometry of noncommutative surfaces. J. Math. Phys. {\bf 49} (2008), 073511.

\bibitem{C01} A. H. Chamseddine, Complexified gravity in noncommutative spaces. Commun. Math. Phys. {\bf 218} (2001), 283-292.

\bibitem{C04} A. H. Chamseddine, $SL(2,\C)$ gravity with a complex vierbein and its noncommutative
extension. Phy. Rev. {\bf D 69} (2004), 024015.

\bibitem{CC96} A. H. Chamseddine, A. Connes,  Universal formula for noncommutative geometry actions: unification of gravity and the
standard model. Phys. Rev. Lett. {\bf 77} (1996), no. 24, 4868-4871.

\bibitem{CC10} A. H. Chamseddine, A. Connes, Noncommutative geometry as a framework for unification of all
fundamental interactions including gravity. Part I. Fortschr. Phys. {\bf 58} (2010), no. 6, 553-600.

\bibitem{C} C.J.S. Clarke, On the global isometric embedding of pesudo-Riemannian manifolds.  Proc. Roy. Soc. Lond. {\bf A  314} (1970), 417-428.

\bibitem{Col} C.D. Collinson, Embeddings of the plane-fronted waves and other space-times. J. Math. Phys. {\bf 9} (1968), 403.

\bibitem{Co} A. Connes, {\em Noncommutative geometry}. Academic Press (1994).

\bibitem{CL} A. Connes, J. Lott, Particle models and noncommutative geometry. Recent advances in field theory (Annecy-le-Vieux, 1990).
Nuclear Phys. B Proc. Suppl. {\bf 18B} (1990), 29-47.

\bibitem{DHLS} L. Dabrowski, P. M. Hajac, G. Landi, P. Siniscalco, Metrics and pairs of left and right connections on bimodules.  J. Math. Phys.
{\bf 37} (1996),  no. 9, 4635-4646.

\bibitem{doC} M. P. do Carmo, {\em Differential geometry of curves and surfaces}. Englewood Cliffs, N.J. : Prentice-Hall (1976).

\bibitem{DKS} B. P. Dolan, Kumar S. Gupta and A. Stern, Noncommutativity and
quantum structure of spacetime. J. Phys. Conf. Ser. {\bf 174}, 012023 (2009), 7 pp.

\bibitem{DFR} S. Doplicher, K. Fredenhagen, J. E. Roberts,
The quantum structure of spacetime at the Planck scale and quantum
fields. Commun. Math. Phys.  {\bf 172} (1995),  187-220.

\bibitem{DN} M. R. Douglas,  N. A. Nekrasov, Noncommutative field theory.
Rev. Mod. Phys.  {\bf 73} (2001), 977-1029.

\bibitem{D} V. Drinfeld, Quasi-Hopf algebras. Leningrad Math. J. {\bf 1} (1990), 1419-1457.

\bibitem{DMMM} M. Dubois-Violette, J. Madore, T. Masson, J. Mourad. On curvature in noncommutative geometry. J. Math. Phys. {\bf
37} (1996), no. 8, 4089-102.

\bibitem{ER} A. Einstein, N. Rosen, On gravitational waves. J. Franklin Inst. {\bf 223} (1937), 43-54.

\bibitem{EK} J. Ehlers, W. Kundt, Exact solutions of the gravitational field equations, {\em Gravitation: an Introduction to Current Research}
(1962), 49-101.

\bibitem{FW} G. Fiore and J. Wess, Full twisted Poincar\'e symmetry and quantum field theory on Moyal-Weyl spaces.
Phys. Rev. {\bf D 75} (2007), 105022.

\bibitem{Fr} A. Friedman, Local isometric embedding of Riemannian manifolds with indefinite metric. J. Math. Mech. {\bf 10} (1961), 625.

\bibitem{F} C. Fronsdal, Completion and embedding of the Schwarzschild solution. Phys. Rev. {\bf 116} (1959), 778.

\bibitem{Ge} M. Gerstenhaber, On the deformation of rings and algebras. Ann. Math. (2) {\bf 79} (1964), 59-103.

\bibitem{GBV}
V. Gayral, J. M. Gracia-Bond\'ia, B. Iochum, T. Sch\"ucker, J. C. V\'arilly, Moyal Planes are Spectral Triples. Commun. Math. Phys. {\bf 246} (2004), 569-623.

\bibitem{GVF} J. M. Gracia-Bond\'ia, J. C. V\'arilly, H. Figueroa,  {\em Elements of noncommutative geometry}. Birkh\"auser Advanced Texts: Basler
Lehrb\"ucher.  Birkh\"auser Boston, Inc., Boston, MA (2001).

\bibitem{Gr} R. E. Greene, Isometric embedding of Riemannian and pseudo-Riemannian manifolds. Memoirs Am. Math. Soc. {\bf 97} (1970).

\bibitem{HW} H. Grosse, R. Wulkenhaar, Renormalisation of $\phi^4$-theory on noncommutative $\R^4$ in the matrix base. Commun. Math. Phys.
{\bf 256} (2005), no. 2, 305-374.

\bibitem{KS} M. Kashiwara and Schapira, Deformation quantization modules I: Finiteness and duality,
arXiv:0802.1245 [math.QA]; Deformation quantization modules II. Hochschild class, arXiv:0809.4309 [math.AG].

\bibitem{K1}
E. Kasner, Finite representation of the solar gravitational field in flat space of six dimensions. Am. J. Math. {\bf 43} (1921),
130.

\bibitem{K2} E. Kasner, The impossibility of Einstein fields immersed in flat space of five dimensions.  Am. J. Math. {\bf 43} (1921), 126.

\bibitem{Kob} A. Kobakhidze, Noncommutative corrections to classical black holes. Phys. Rev. {\bf D 79}, 047701 (2009), 3 pp.

\bibitem{Ko} M. Kontsevich, Deformation quantization of Poisson manifolds. Lett. Math. Phys.  {\bf 66} (2003),  no. 3, 157-216.

\bibitem{MM} J. Madore, J. Mourad, Quantum space-time and classical gravity. J. Math. Phys. {\bf 39} (1998), no. 1, 423-442.

\bibitem{M05} S. Majid, Noncommutative Riemannian and spin geometry of the standard $q$-sphere.
Commun. Math. Phys. {\bf 256} (2005), no. 2, 255-285.

\bibitem{MH} F. Muller-Hoissen, Noncommutative geometries and gravity. AIP Conference Proceedings {\bf 977} (2008), 12-29.

\bibitem{MRS} S. Minwalla, M. Van Raamsdonk, N. Seiberg, Noncommutative perturbative dynamics. J. High Energy Phys. {\bf 0002}
(2000), no. 2, 020.

\bibitem{N} J. Nash, The imbedding problem for Riemannian manifolds. Ann. Math. {\bf 63} (1956), 20-63.

\bibitem{OS} J.R. Oppenheimer, H. Snyder, On continued gravitational contraction, Phys. Rev. {\bf 56} (1939), 455-459.

\bibitem{Pe} R. Penrose, Any space-time has a plane wave as a limit,
{\em Differential geometry and relativity}, 271-275, Math. Phys.
Appl. Math., Vol 3, Reidel, Dordrecht (1976).

\bibitem{Ro} N. Rosen, Plane polarized waves in the general theory of relativity. Phys. Z. {\bf 12} (1937), 366-372.

\bibitem{SW} N. Seiberg, E. Witten, String theory and noncommutative geometry. J. High Energy Phys. {\bf 9909} (1999), 032.

\bibitem{Sn} H. S. Snyder, Quantized Space-time, Phys. Rev. {\bf 71} (1947), 38-41.

\bibitem{SV} J. T. Stafford, M. Van den Bergh, Noncommutative curves and
noncommutative surfaces. Bull. Amer. Math. Soc. (N.S.) {\bf 38} (2001), 171-216.

\bibitem{St} H. Steinacker, Emergent gravity and noncommutative branes from Yang-Mills matrix models.
 Nucl. Physics {\bf B 810} (2009), 1--39.

\bibitem{SWXZZ} W. Sun, D. Wang, N. Xie, R. B. Zhang, X. Zhang,
Gravitational collapse of spherically symmetric stars in noncommutative general relativity. Eur. Phys. J. {\bf C 69} (2010), no. 1-2, 271-279.

\bibitem{Sz} R. J. Szabo, Symmetry, gravity and noncommutativity. Class. Quant. Grav. {\bf 23} (2006) R199.

\bibitem{T} R.C. Tolman, Static solutions of Einstein's field equations for spheres of fluid, Phys. Rev.
{\bf 55} (1939), 364-373.

\bibitem{W} R. Wald, General relativity, Chicago, IL: University of Chicago Press (1984).

\bibitem{WZZ1} D. Wang, R. B. Zhang, X. Zhang, Quantum deformations of Schwarzschild and Schwarzschild-de Sitter
spacetimes. Class. Quantum Grav. {\bf 26}, 085014 (2009), 14 pp.

\bibitem{WZZ2} D. Wang, R. B. Zhang, X. Zhang, Exact solutions of noncommutative vacuum Einstein field equations and plane-fronted
gravitational waves. Eur. Phys. J. {\bf C 64} (2009), no. 3, 439-444.

\bibitem{Y} C. N. Yang, On Quantized Space-time, Phys. Rev. {\bf 72} (1947), 874.

\bibitem{ZZ} R. B. Zhang, X. Zhang, Projective module description of embedded noncommutative spaces. Rev. Math. Phys. {\bf 22} (2010), no. 5,
507-531.

\end{thebibliography}
\end{document}